\documentclass[journal,12pt,onecolumn,draftclsnofoot,]{IEEEtran}

\normalsize

\usepackage{cite}

\ifCLASSINFOpdf
\else
\fi

\usepackage{algorithm}
\usepackage[noend]{algorithmic}	
\usepackage{amsthm}
\newtheorem{theorem}{Theorem}
\newtheorem{lem}[theorem]{Lemma}
\newtheorem{prop}{Proposition}

\theoremstyle{definition}

\newtheorem{rem}{Remark}
\usepackage{amsmath,amssymb}
\DeclareMathOperator*{\argmin}{argmin}
\DeclareMathOperator*{\argmax}{argmax}

\usepackage{lipsum}
\usepackage{graphicx}
\ifCLASSOPTIONcompsoc
    \usepackage[caption=false, font=normalsize, labelfont=sf, textfont=sf]{subfig}
\else
\usepackage[caption=false, font=footnotesize]{subfig}
\fi

\usepackage{url}
\usepackage{verbatim}
\usepackage{array}

\DeclareMathOperator{\Exp}{E}
\newcommand*{\Break}{\textbf{break}}

\usepackage{float}    
\usepackage{graphicx}	
\providecommand{\myceil}[1]{\left \lceil #1 \right \rceil }

\usepackage{xcolor}

\begin{document}
%
\title{Age-Aware Dynamic Frame Slotted ALOHA for Machine-Type Communications}
%
%
%

\author{Masoumeh~Moradian,
        Aresh~Dadlani,~\IEEEmembership{Senior~Member,~IEEE,}
        Ahmad~Khonsari,
        and~Hina~Tabassum,~\IEEEmembership{Senior~Member,~IEEE}
\thanks{M. Moradian is with the School of Computer Engineering, K. N. Toosi University of Technology, Tehran, Iran, and with the School of Computer Science, Institute for Research in Fundamental Sciences, Tehran, Iran (e-mail: mmoradian@kntu.ac.ir).}
\thanks{A. Khonsari is with the School of Electrical and Computer Engineering, University of Tehran, Iran, and also with the School of Computer Science, Institute for Research in Fundamental Sciences, Tehran, Iran (e-mail: a\_khonsari@ut.ac.ir).}
\thanks{A. Dadlani is with the Department of Computing Science, University of Alberta, Edmonton, AB T6G 2E8, Canada, and also with the Department of Electrical and Computer Engineering, Nazarbayev University, Astana 010000, Kazakhstan (e-mail: aresh@ualberta.ca).}
\thanks{H. Tabassum is with the Department of Electrical Engineering and Computer Science, York University, Toronto, ON M3J 1P3, Canada (e-mail: hinat@yorku.ca).}
}

%
%

\markboth{IEEE Transactions on Communications}%
{Submitted paper}
%



\maketitle
\vspace{-0.3em}

\begin{abstract}
\fontdimen2\font=0.54ex
Information aging has gained prominence in characterizing communication protocols for timely remote estimation and control applications. This work proposes an Age of Information (AoI)-aware threshold-based dynamic frame slotted ALOHA (T-DFSA) for contention resolution in random access machine-type communication networks. Unlike conventional DFSA that maximizes the throughput in each frame, the frame length and age-gain threshold in T-DFSA are determined to minimize the normalized average AoI reduction of the network in each frame. At the start of each frame in the proposed protocol, the common Access Point (AP) stores an estimate of the age-gain distribution of a typical node. Depending on the observed status of the slots, age-gains of successful nodes, and maximum available AoI, the AP adjusts its estimation in each frame. The maximum available AoI is exploited to derive the maximum possible age-gain at each frame and thus, to avoid overestimating the age-gain threshold, which may render T-DFSA unstable. Numerical results validate our theoretical analysis and demonstrate the effectiveness of the proposed T-DFSA compared to the existing optimal frame slotted ALOHA, threshold-ALOHA, and age-based thinning protocols in a considerable range of update generation rates.
\end{abstract}

\vspace{-0.3em}

%
\IEEEpeerreviewmaketitle

\section{Introduction}
\label{sec1}
\fontdimen2\font=0.54ex
\IEEEPARstart{B}{y} virtue of the ubiquitous connection offered by emerging wireless technologies and the vast~proliferation of smart devices, the Internet of Things (IoT) and Machine-Type Communications (MTC) have revolutionized the way information is shared and processed \cite{Rita2016}. Making precise control and monitoring decisions in such real-time networked systems depends critically on the timely reception of information at the destination node. The application domains of the IoT devices essentially~determine how long the information aggregated at a common Access Point (AP) may still be regarded as fresh \cite{Elmagid2019, Chiariotti2022}. To quantify the timeliness of status updates in such large-scale systems, the Age of Information (AoI) has been widely adopted as a key destination-centric performance metric \cite{Kaul2012}. Formally, the AoI at any given time $t$ is defined as the random process $\Delta(t) \triangleq t - u(t)$, where~$u(t)$~is the generation time of the most recent status~update received at the monitoring node~\cite{Kumar2023}. Previous studies have shown that maximizing throughput or reducing delay does not necessarily ensure a low AoI~\cite{Yates2021}.

Besides the stringent temporal constraint, the AoI in IoT~networks is also influenced by the underlying channel access scheduling policy for contending nodes that intermittently transmit short~status packets~\cite{Ghavimi2015}. Though Time Division Multiple Access (TDMA)-based access protocols are power-efficient~for individual nodes, they are not suitable for decentralized applications due to their low spectrum efficiency and lack of scalability.~For this reason, random access protocols such as Slotted ALOHA (SA) and Framed Slotted ALOHA (FSA) are deemed more effective in capturing~the nuances of information aging in timely communications with bursty traffic \cite{Atabay2020}. Existing age-aware slot-based protocols leverage threshold-based structures to limit contention to nodes with higher age-gains, where the age-gain of a node in each slot is defined as the difference~between the ages of the most recent updates available at the AP and the node~\cite{Tahir2021,Chen2022,Ahmetoglu2022}. Higher age-gains indicate that the most recent update at the node was generated a long time after its last~successful transmission~\cite{Chen2022}. While threshold-based~SA protocols significantly reduce the Average AoI (AAoI) compared to conventional protocols, they may not be energy-efficient~due to their reliance on contention alone.~In contrast, hybrid protocols such as FSA strike a balance between contention-based and contention-free approaches, thus making them well-suited for settings with low-power nodes~\cite{Ajinkya2015,Yu2021}. By centrally controlling node access to the medium, hybrid protocols mitigate excessive collisions and offer improved energy efficiency. In particular, FSA has long been the \emph{de facto} protocol in Machine-to-Machine (M2M) and RFID communications since it restricts node access to a limited number of slots within frames, thus effectively reducing collision probability~\cite{Ajinkya2015,Yu2021,wang2023age,yue2023age,huang2023distributed,Saha2021, Andrea2021, Huang2022, Li2022}.

Various FSA variants have been subject to recent investigations within the context of AoI. These studies have probed~the performance of FSA in scenarios with frame-synchronous~and frame-asynchronous users~\cite{wang2023age}, multiple source-destination pairs \cite{yue2023age}, and heterogeneous users~\cite{huang2023distributed}. Furthermore, exploration of the potential for age improvement has been conducted with a more recent variant of FSA, known as Irregular FSA (IRSA)~\cite{Saha2021, Andrea2021, Huang2022, Li2022,huang2023distributed}. In IRSA, the frame length is assumed to~be fixed, and each node transmits $d$ replicas of a single update within the frame, with the value of $d$ determined by a common distribution across all nodes. The AP uses Successive Interference Cancellation (SIC) at the end of the frame to detect the~transmitted updates. An issue common to FSA protocols with fixed frame lengths, including IRSA, is that the frame size may not align with the number of active (contending) nodes, which is often unknown at the start of frames, particularly in networks characterized~by stochastic update generation. This mismatch can result in throughput degradation or network instability, as the frame length may be significantly larger or smaller than the number of active nodes. A simplified version of IRSA, referred to as Frameless Slotted ALOHA (FA), which adopts a binomial distribution for $d$~\cite{stefanovic2012frameless}, is studied under an enforcing policy that restricts access to only nodes with non-zero age-gains in each frame~\cite{huang2023age}. The FA with dynamic frame has~also been~scrutinized for age improvement, wherein the AP extends the frame duration until~all updates are recovered using SIC~\cite{munari2023dynamic}.~To address the instability issue in FSA,~the Dynamic FSA (DFSA) has been proposed~\cite{Frits1983, Su2019}. In DFSA, the AP dynamically adjusts the frame length based~on estimations of backlogged nodes~in each frame, aiming to optimize throughput. Consequently, DFSA entails estimating the number of backlogged nodes and broadcasting the frame length at the onset of each frame. In light of the favorable characteristics of DFSA, such as suitability for low-power scenarios and the potential for dynamic decision-making at the AP to maintain data freshness, our work concentrates on refining its estimation and broadcast processes for age improvement. In particular, we incorporate AoI-related information of both successful nodes and the AP in the estimation process, and introduce an age-gain threshold in addition to the frame size in the broadcast process. Additionally, we depart from the assumptions of prior studies by allowing nodes to retransmit updates, a feature particularly valuable in settings where nodes cannot generate updates at will.\vspace{-0.2em}

\subsection{Related Work}
\label{subsec1.A}
\fontdimen2\font=0.54ex
AoI performance of scheduling policies in shared channels has garnered much research attention in recent years. Concerning decentralized access protocols, SA has been analyzed for periodic~status updates \cite{Bae2022}, stochastic arrivals \cite{Kadota2021}, and unreliable channels \cite{Yates2017, Munari2021}. The collision-resolution~SA protocol in \cite{Pan2022} reduces the AAoI by resolving all collisions in each access period.~Integrating~age-aware transmissions into SA protocols can further improve their age efficiency. By exploiting~the optimal age threshold and transmission probability, \cite{Tahir2021} shows that the optimal AAoI under Threshold-ALOHA (TA) grows linearly with the network size, yielding a multiplicative gain of $1.4169$. The~TA protocol is modified in \cite{Ahmetoglu2022} to include a mini-slot at the start of each slot, within which active users contendwith certain probabilities and back off in case of collisions. The Stationary Age-based Thinning (SAT) and the Adaptive Age-based Thinning (AAT) protocols for stochastic arrivals are proposed in \cite{Chen2022}, where each node transmits its latest packet according to SA only if its age-gain exceeds a certain threshold. This threshold is estimated dynamically by monitoring ACK messages or determined statically, leading to similar scaling results as in \cite{Tahir2021}. In particular, in AAT, each node estimates the age-gain distribution at the beginning of each slot using the collision status of~the previous slot. Compared to AAT and SAT, the proposed T-DFSA protocol updates the estimated age-gain distribution by leveraging the successfully received age-gains in the previous frame. Moreover, in T-DFSA, the~AP is responsible for updating estimations in each slot and determining the threshold, thus making it more suitable for low-power nodes.

Among the FSA-based protocols, \cite{Wang2022} investigates the AoI performance of FSA, where each~frame consists of a fixed~number of data slots and one reservation slot. The AAoI in the FSA structure, where users compete using a pool of pilots~in each frame, has been explored and optimized in \cite{Yang2022}. The design and stability analysis of an FSA-based protocol in the presence of new arrivals and an unknown node population are provided in \cite{Yu2021}. In \cite{wang2023age}, the authors analyze the performances of frame-synchronous and frame-asynchronous FSA. The IRSA protocol is a recent adaptation of FSA~that allows repeated update transmissions in a frame and uses SIC to offset the latency cost of framed channel access. The work in \cite{Andrea2021} examines~the AAoI and age-violation likelihood to optimize the IRSA repetition distribution and frame size, while analytical results for the optimal normalized channel traffic and repetition distribution are derived in \cite{Saha2021}.~To optimize the number of transmissions in IRSA, \cite{Li2022} presents a Q-learning solution for energy-harvesting IoT networks. In \cite{Huang2022}, the AAoI and packet loss rate expressions for a graph-based spatially coupled IRSA protocol combined with the pseudo-random access method and coupled frames are obtained. In~\cite{yue2023age}, the AoI of IRSA is analyzed using stochastic geometry in the presence of source-destination pairs. In \cite{huang2023distributed}, users are categorized into heterogeneous groups, and the nodes within each group transmit under IRSA with specific probabilities. FA protocols differ from IRSA in that active nodes transmit in each slot of the frame with specific~probabilities, and the AP utilizes SIC at the end of the frame to recover the updates. The authors of \cite{huang2023age} analyze the performance of FA, where successful nodes in one frame are not allowed to transmit in the next. The FA protocol in \cite{munari2023dynamic} extends the frame length until all updates from active nodes are recovered, and various performance metrics, including AoI, are analyzed. However, they do not account for retransmissions in their scenario.

The throughput of the standard DFSA protocol is optimized by setting the frame length equal to the number of backlogged nodes. Estimating the number of backlogged nodes is crucial, especially~in scenarios with stochastic packet arrivals. Existing efforts primarily focus on exploiting the observed status of slots within a single frame to minimize estimation errors within that frame. Algorithms such as the Vogt algorithm \cite{Vogt2002}, the maximum a posteriori probability rule \cite{Chen2009}, and the Schoute algorithm \cite{Frits1983} have been proposed for this purpose. Although DFSA has been scrutinized for maximizing throughput, its performance in relation to the new AoI requirement has not been investigated. This AoI-incorporated setup is ideally pertinent for remote estimation and control of observed processes in power-limited IoT networks.  To the best of our knowledge, this work presents the first proposal for improving the DFSA protocol that benefits from a threshold-based approach to minimize age, specifically considering updates generated in a non-deterministic manner.\vspace{-0.3em}

\subsection{Contributions}
\label{subsec1.B}
\fontdimen2\font=0.54ex
In this work, we focus on laying the theoretical foundations for the \emph{age-aware} design of a novel \emph{threshold-based} DFSA (T-DFSA) policy for networks in which the status updates~\emph{stochastically} generated by nodes are aggregated at the AP. The main contributions of this paper are as follows:
\begin{itemize}
	\item We propose the ideal T-DFSA, in which the common AP broadcasts the frame length and the age-gain threshold at the start of each frame to identify the nodes eligible to transmit in that frame. The Average AoI Reduction (AAR) per slot is minimized by determining the threshold and frame size of each frame. We establish that the AAR is minimized if only the nodes with the highest age-gains transmit in each frame and the frame size is equal to the number of such nodes, presuming that the age-gains are known at the start of each frame.
	\item We then propose the practical T-DFSA, where the AP, being unaware of the exact age-gains of the nodes, maintains~an estimate of the age-gain probability mass function associated with a typical node at the start of each frame and updates it across frames. As such, the AP implements T-DFSA in four steps within each frame. The frame length and the threshold are specified and broadcast using the current estimation in the first step. Observations of the AP within the frame, such as the slot states and the age-gains of the successful nodes, are used to update estimates in the second step. In the third and fourth steps, the estimates are further enhanced using the known node generation rates and the maximum achieved AoI, respectively.
	\item Though the frame lengths in T-DFSA are generally small, we demonstrate that it outperforms the optimal FSA (i.e., FSA with optimal fixed frame length) in Average Age-Gain (AAG), even if adjustments to the minimum frame length are based on the AP broadcast frequency~requirements. Moreover, operating at its optimal point, T-DFSA outperforms existing SA-based protocols including TA,~AAT, and SAT in a considerable range of generation rates. Our~numerical results also show that T-DFSA remains stable almost surely if the initial conditions of the nodes are set to be sufficiently diverse.
\end{itemize}\vspace{-0.6em}

\subsection{Organization}
\label{subsec1.C}
\fontdimen2\font=0.54ex
The rest of the paper is organized as follows. Section~\ref{sec2}~describes the system model. The ideal and practical models of T-DFSA, where the AP is, respectively, aware of and unaware of the age-gains of the source nodes, are introduced in Section~\ref{sec3}. The mechanisms of the practical T-DFSA are further detailed in Section~\ref{sec4}, which involve updating the estimations for the following frame and using the estimates from each frame to~calculate the threshold and frame length. Section~\ref{sec5} analyzes the complexity of stable T-DFSA. Numerical results~are presented~in Section~\ref{sec6} and conclusions are drawn in Section~\ref{sec7}.\vspace{-0.4em}

\section{System Model and Notations}
\label{sec2}
\fontdimen2\font=0.54ex
We consider the network in \figurename~\ref{fig0} where $N$ time-synchronous source nodes attempt to send their status~updates to a common AP via an ideal shared wireless channel using the T-DFSA access protocol. The network operates in slotted time, where each frame consists of multiple slots. At the onset of each frame $t$, the AP broadcasts the frame length $w_t$, which denotes the number of slots in the frame, and a threshold $\Gamma_t$. Subsequently, backlogged nodes with an age-gain equal to or greater than $\Gamma_t$ become active. These nodes choose a slot uniformly at~random within~the current frame and transmit their updates with probability one. We assume that if multiple nodes transmit in the same slot, the AP is unable to decode any transmission successfully and thus, a \emph{collision} occurs. On the other hand, if AP receives no signal, it recognizes an \emph{empty} slot with no transmission. Finally, a \emph{singleton} slot, representing a single node transmission translates to successful transmission due to an ideal channel\footnote{We do not consider SIC here since the nodes only transmit once within the frame.~Transmitting multiple replicas within each frame would necessitate specifying distributions for the number of replicas at various frame sizes, which is complicated and has not yet been reported for use in DFSA.}. At the start of frame $t$, the AP broadcasts the slot statuses from the previous frame $t-1$, along with $w_t$ and $\Gamma_t$, enabling successful nodes to discard their sent updates. We assume that the time required for broadcasting these information is negligible compared to frame size $w_t$. All updates are considered equally important and are assumed to be generated at the start of a slot with probability $\lambda$. Each node retains only the most recent update and retransmits it after activation until it is either successfully sent or replaced by a fresher update. Table~\ref{tab_notation} summarizes the main notations.
\begin{figure}[!t]
	\centering
	\includegraphics[width=0.5\columnwidth]{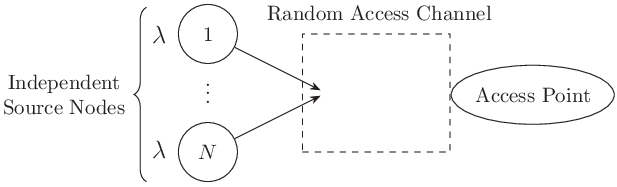}
	\vspace{-0.6em}
	\caption{System model with $N$ source nodes generating updates with~probability $\lambda$ and transmitting them to a common AP over a shared medium.}
	\label{fig0}
\end{figure}
\begin{table}[!t]
\renewcommand{\arraystretch}{1.2}
\caption{Main Notations and Definitions}
\label{tab_notation}
\centering
\begin{tabular*}{\linewidth}{@{\extracolsep{\fill}}l|l}
\hline
\textbf{Notation} & \textbf{Description} \\
\hline\hline
$N$ & Number of independent source nodes. \\
$w_t$ & Length of frame $t$. \\
$\Gamma_t$ & Age-gain threshold in frame $t$. \\
$x_t^i$ & AoI of node $i$ at the onset of frame $t$. \\
$y_t^i$ & AoI of node $i$ at the AP at the onset of frame $t$. \\
$g_t^i$ & Age-gain of node $i$ at the onset of frame $t$. \\
$r_t^i (R_t^i)$ & (Expected) AoI reduction of node $i$ in frame $t$. \\
$\bar{R}_t$ & Average AoI reduction in frame $t$. \\
$P_t^s$ & Successful transmission probability of an active node in\vspace{-0.2em}\\& frame $t$. \\
$n_t^a$ & Number of nodes with age-gain of $a$ at the onset of frame~$t$. \\
$n_{t,s}^a$ & Number of successful nodes with age-gain of $a$ in frame $t$. \\
$\hat{f}_t^a (\hat{f}_{t^+}^a)$ & Estimated probability of a node with age-gain of $a$ at the\vspace{-0.2em}\\& onset (end) of frame $t$. \\
$\hat{n}_t^a (=\! N\hat{f}_t^a)\!\!$ & Estimated average number of nodes with age-gain of $a$ at the\vspace{-0.2em}\\& onset of frame $t$. \\
$\hat{m}_t^a$ & Posterior estimated mean number of nodes with age-gain of\vspace{-0.2em}\\& $a \geq \Gamma_t$ at the onset of frame $t$. \\
$\mathcal{G}_t (\hat{\mathcal{G}}_t)$ & Set (estimated set) of distinct age-gains of all nodes at the\vspace{-0.2em}\\& onset of frame $t$. \\
$\mathcal{A}_t (\hat{\mathcal{A}}_t)$ & Set (posterior estimated set) of distinct age-gains of active\vspace{-0.2em}\\& nodes at the onset of frame $t$. \\
$\mathcal{S}_t$ & Set of distinct age-gains of successful nodes in frame $t$. \\
\hline
\end{tabular*}
\end{table}


To index the slots within a generic frame $t$, we define slot $k$ as the time interval $[k, k+1)$, and $k_t$ as the first slot in that frame. Therefore, the time slots within frame $t$ are represented by the set $\mathcal{K}_t = \{k_t, k_t + 1, \ldots, k_t + w_t - 1\}$. Accordingly, the AoI of node $i$ at the start of frame $t$, denoted by $x_t^i$, evolves over time as follows:\vspace{-0.2em}
\begin{align}
    x_{t+1}^i =
    \begin{cases}
        x_t^i + w_t, & \text{if $u^i(k) = 0, \forall k \in \mathcal{K}_t$}\\[2pt]
        k_{t+1}-k, & \text{if $\exists k \in \mathcal{K}_t\!: u^i(k) = 1$,}\\
        &\quad\text{$u^i(j)=0$ for $j \in \mathcal{K}_t, j > k$},
    \end{cases}
    \label{eq1}
    \vspace{-0.2em}
\end{align}
where $u^i(k) = 1$ if node $i$ generates an update in time slot $k$ and zero otherwise. Assuming that the AP decodes the packets at the end of the frames, the AoI of the update packet received at the AP from node $i$ can be expressed as:\vspace{-0.2em}
\begin{equation}
  y_{t+1}^i =
  \begin{cases}
    x_t^i + w_t, & \text{ if $d_t^i = 1$} \\[2pt]
    y_t^i + w_t, & \text{ if $d_t^i = 0$}, 
  \end{cases}
  \label{eq2}
  \vspace{-0.2em}
\end{equation}
where $d_t^i$ takes the value of $1$ if node $i$ transmits successfully in frame $t$ and $0$ otherwise. Realizations of $x_t^i$ and $y_t^i$ along with the collision and singleton slots for node $i$ are exemplified in \figurename~\ref{fig1}. Based~on the definitions in \eqref{eq1} and \eqref{eq2}, the age-gain $g_t^i$ and age reduction $r_t^i$ of node $i$ in frame $t$ can be written as:\vspace{-0.3em}
\begin{equation}
   \begin{cases}
    g_t^i = y_t^i - x_t^i, \\[2pt]
    r_t^i = y_t^i - y_{t+1}^i. 
  \end{cases}
  \label{eq3}
  \vspace{-0.2em}
\end{equation}

Since the backlogged nodes in T-DFSA are determined by non-zero age-gains, it is noteworthy~that $x_t^i$ in \eqref{eq1} is defined regardless of whether node $i$ is backlogged or not, and has no bearing on our approach. Importantly, the nodes that satisfy the condition $g_t^i \geq \Gamma_t$ are the active nodes in frame $t$.
\begin{figure}[!t]
	\centering
	\includegraphics[width=0.5\columnwidth]{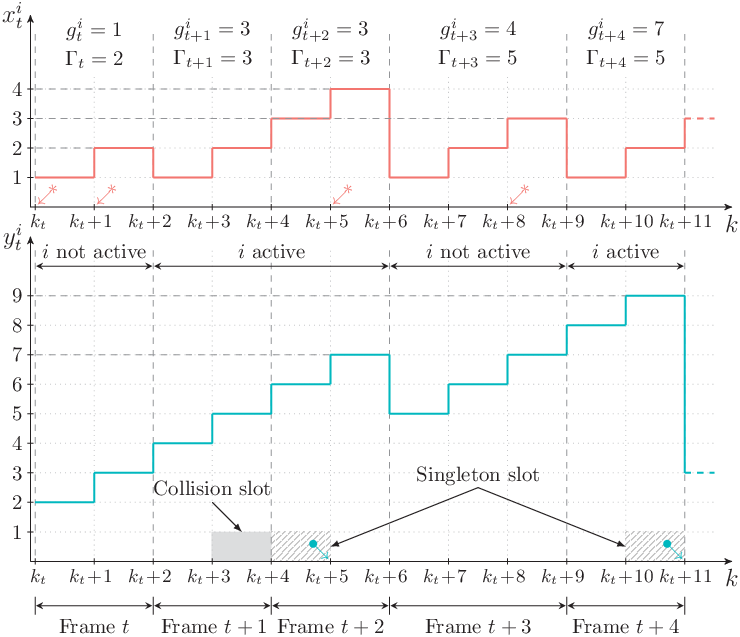}
	\vspace{-0.3em}
	\caption{An example of $x_t^i$ and $y_t^i$ evolving over time under the T-DFSA~protocol. The red arrows (with asterisk tails) represent instances when update packets are generated at node $i$, whereas the blue arrows (with bullet tails) represent instances when they are received at the AP.}
	\label{fig1}
	\vspace{-0.2em}
\end{figure}

It is important to note that in T-DFSA, the selection of active nodes, i.e., the nodes allowed to~transmit in frame $t$, is contingent upon the threshold signaled by the AP.~Furthermore, in T-DFSA, the values of $w_t$ and $\Gamma_t$ are obtained such that the following AAR within frame $t$ is minimized:
\begin{equation}
	\bar{R}_t = \frac{\sum_{i=1}^N \Exp[r_t^i]}{N w_t},
	\label{eq4}
\end{equation}
where $\Exp[\cdot]$ denotes the expectation operator. The rationale behind the division of the total expected age reduction of the network, expressed as $\sum_{i=1}^N \Exp[r_t^i]$, by the frame duration $w_t$ in \eqref{eq4} is to achieve greater age reductions in smaller frames. In fact, $\bar{R}_t$ represents the AAR per slot. 
Additionally, the throughput~in~each frame, denoted by $U_t$, is defined as the average number~of~transmitted updates per slot and is given as:
\begin{equation}
	{U}_t = \frac{\sum_{i:g^i_t \geq \Gamma_t } \Exp[d_t^i]}{ w_t}= \frac{\sum_{i:g^i_t \geq \Gamma_t } \Pr\{d^i_t=1\}}{ w_t},
	\label{eq:thr}
\end{equation}
where $\Exp[d_t^i]$ $(= \Pr\{d^i_t \!=\! 1\})$ is the average number of successfully transmitted updates by node $i$ in frame $t$.

In the section that follows, we first introduce the ideal~T-DFSA in which age-gains are assumed to be known at the~start of the frames. The practical T-DFSA is then proposed, where~the~distribution of $g_t^i$ is estimated by the AP.\vspace{-0.4em}

\section{Proposed T-DFSA Protocol}
\label{sec3}
\fontdimen2\font=0.54ex
In this section, we first gauge the behavior of the hypothetical model in which the age-gain of~each node is known to the AP at the onset of the frames, before delving into the AoI analysis of its practical counterpart.\vspace{-0.4em}

\subsection{Ideal T-DFSA}
\label{sec3.1}
\fontdimen2\font=0.54ex
To facilitate our analysis, we introduce $R_t^i = \Exp[r_t^i | g_t^i]$ to denote the expected age reduction of node $i$ in frame $t$ under ideal T-DFSA. Our derivation of $R_t^i$ begins by considering the scenario~where node $i$ is active in frame $t$ (i.e., $g_t^i \geq \Gamma_t$). Assuming successful transmission by node $i$ in frame $t$, the age reduction of node $i$, denoted by $r_t^i$, can be written based on \eqref{eq2} and \eqref{eq3} as follows:\vspace{-0.1em}
\begin{equation}
	r_t^i = y_t^i - x_t^i - w_t = g_t^i - w_t.
	\label{eq5}
    \vspace{-0.1em}
\end{equation}
On the other hand, if node $i$ is active but fails to transmit~successfully in frame $t$, the value of $y_t^i$ increases by $w_t$, resulting in $y_{t+1}^i = y_{t}^i + w_t$. Using \eqref{eq3}, we then have:
\begin{equation}
	r_t^i = y_{t}^i - y_{t+1}^i = y_{t}^i -y_{t}^i-w_t=-w_t.
	\label{eq5n}
\end{equation}

Let us denote $P_t^s$ as the successful transmission probability of an active node in frame $t$, given by $P_t^s \!=\! \Pr\{d_t^i \!=\!1\}$.~Note that this probability is identical for all active nodes, as they exhibit symmetry when transmitting their packets within the frame. From \eqref{eq5} and \eqref{eq5n}, we can express the expected age reduction~of an active node $i$ in frame $t$ as follows:
\begin{equation}
    R_t^i = P_t^s (g_t^i - w_t) + (1 - P_t^s)(-w_t)=P_t^s g_t^i - w_t.
\end{equation}

If node $i$ remains inactive throughout the frame, then its age increases by $w_t$ with~probability one, leading to $y_{t+1}^i = y_{t}^i + w_t$. Consequently, from \eqref{eq3}, the age reduction $r_t^i$ can be~expressed~as $r_t^i = -w_t$. As a result of this, we have $R_t^i = r_t^i = -w_t$. By consolidating all these considerations, the expression for $R_t^i$ is:
\begin{equation}
  R_t^i =
  \begin{cases}
    P_t^s g_t^i - w_t, & \text{ if $g_t^i \geq \Gamma_t$} \\[2pt]
    -w_t, & \text{ if otherwise}. \\
  \end{cases}
  \label{eq6}
\end{equation}

We now define $n_t^a$ as the number of nodes with an age-gain of $a$ at the start of frame $t$. Then,~the total number of active nodes at the beginning of frame $t$, denoted as $n_t$, can be expressed as $n_t \triangleq \sum_{a = \Gamma_t}^{M_t} n_t^a$, where $M_t = \max_i g_t^i$ is the maximum available age-gain at the beginning of frame $t$. Using the definition of $n_t$, the expression for $P_t^s$ can be derived as:\vspace{-0.2em}
\begin{equation}
  P_t^s = \left(1 - \frac{1}{w_t}\right)^{n_t - 1}.
  \label{eq7}
\end{equation}
In fact, 
$P_t^s$ is the~probability that none of the other $n_t-1$ active nodes choose the same slot as node $i$ and can be expressed as $(1 - 1/{w_t})^{n_t - 1}$.

Given that the age-gains are known in ideal T-DFSA, we can establish that $\Exp[r_t^i] = \Exp[r_t^i|g^i_t]$ and thus, the definition of AAR in \eqref{eq4} can be expressed as $\sum_{i=1}^N \Exp[r_t^i|g^i_t]/(N w_t)$. Taking into account~that $\Exp[r_t^i|g^i_t] = R^i_t$ and using \eqref{eq6}, the AAR in frame $t$ under ideal T-DFSA can be obtained as:
\allowdisplaybreaks
\begin{align}
	\bar{R}_t 
        &= \frac{1}{N w_t\!} \sum_{i=1}^N \!R_t^i 
	            = \frac{1}{Nw_t\!}\!\left(\sum_{i: g_t^i \geq \Gamma_t}\!\!\! (P^s_t g_t^i-w_t) +\! \sum_{i: g_t^i < \Gamma_t}\!\!\! (- w_t)\!\right)  \nonumber \\
        &= \frac{1}{Nw_t\!}\!\left(\!P^s_t\!\sum_{i: g_t^i \geq \Gamma_t}\!\!\! g_t^i -\! \sum_{i: g_t^i \geq \Gamma_t}\!\!\! w_t -\! \sum_{i: g_t^i < \Gamma_t}\!\!\! w_t\right) \nonumber \\
        &= \frac{1}{Nw_t\!}\left(\!P^s_t\!\sum_{i: g_t^i \geq \Gamma_t}\!\!\! g_t^i - \sum_{i=1}^N w_t\right) =\frac{P_t^s}{N w_t\!} \sum_{i: g_t^i \geq \Gamma_t}\!\!\! g_t^i - 1 \nonumber \\
        &= \frac{1}{N w_t\!}\left(1 - \frac{1}{w_t}\right)^{\sum_{a = \Gamma_t}^{M_t} n_t^a - 1} \sum_{a = \Gamma_t}^{M_t}\! a\, n_t^a - 1.
    \label{eq8}
\end{align}
The final expression in \eqref{eq8} is derived using \eqref{eq7} along with the equalities $n_t = \sum_{a = \Gamma_t}^{M_t} n_t^a$ and $\sum_{i:g_t^i \geq \Gamma_t} g_t^i = \sum_{a=\Gamma_t}^{M_t} a\, n_t^a$. Given the values of $g_t^i$ for all $i$, the following lemma indicates the~optimal $\Gamma_t$ and $w_t$ values that maximize AAR in frame $t$.
\begin{lem}
	\label{lemma1}
	In ideal T-DFSA, setting $\Gamma_t = M_t$ and $w_t = n_t^M$ maximizes $\bar{R}_t$.
\end{lem}
\begin{proof}
	See Appendix~\ref{appA}.
\end{proof}

From the above lemma, it becomes evident that in ideal T-DFSA, the AP allows nodes with the highest age-gain to transmit within a frame, while setting the frame length equal to the number of nodes with the highest age-gain. Considering the definition of $P^s_t$ and using \eqref{eq7}, the throughput in \eqref{eq:thr} can thus be written as:
\begin{equation}
    {U}_t = \frac{n_t P^s_t}{ w_t}=\frac{n_t}{w_t}\left(1-\frac{1}{w_t}\right)^{n_t-1}.
	\label{eq:thr1}
\end{equation}
It is worth noting that in \eqref{eq:thr1}, $U_t$ is the throughput of a DFSA protocol during frame $t$, given that there are $n_t$ active nodes permitted to transmit in this frame under the random access strategy outlined in Section~\ref{sec2}.~In conventional DFSA, where all nodes are identical, any backlogged~node is expected to be active.~Thus, $n_t$ indicates the number of backlogged nodes.~However, in T-DFSA,~active nodes are defined as those with the highest age-gain.~In both cases, the throughput adheres~to \eqref{eq:thr1} and is maximized when $w_t = n_t$ (see Appendix~\ref{appA} for the~proof). Therefore, as stated in Lemma~\ref{lemma1}, the AAR-optimal T-DFSA, which sets $w_t = n^{M_t}_t$, maximizes the throughput of T-DFSA as well.~In fact,~the AAR-optimal policy maximizes the number of successfully transmitted packets, but it exclusively allows nodes with the~maximum age-gain to transmit.\vspace{-0.4em}

In the case when $\lambda = 1$, the age-gains are known at the AP and can be expressed as $g_t^i = y_t^i - 1$, since $x_t^i = 1$ for all $i$. This implies that the nodes always have a fresh update at the beginning of each frame. By applying Lemma~\ref{lemma1} in this scenario, we can conclude that only the nodes satisfying $y_t^i = \max_i y_t^i$ are allowed to transmit, and the frame length is determined by the number of these nodes. It is not viable for the AP to acquire the age-gains of all source nodes when $\lambda < 1$, especially~in large-scale IoT and M2M communications. Therefore, in the following sub-section, we present a~practical T-DFSA protocol where the AP estimates the average values of $n_t^a$ using observations from the preceding frame before deciding on $w_t$ and $\Gamma_t$ values.\vspace{-0.5em}

\subsection{Practical T-DFSA}
\label{sec3.2}
\fontdimen2\font=0.54ex
The AP in practical T-DFSA stores an estimate of the age-gain distribution of a typical source~node at the start of frame $t$. Accordingly, the estimated probability that a node has an age-gain of $a$ at the onset of frame $t$ is denoted by $\hat{f}_t^a$. The AP also keeps an estimate of the set of distinct age-gains,~given by $\hat{\mathcal{G}}_t$, at the start of frame $t$ where for all $a \in \hat{\mathcal{G}}_t$, we have $\hat{f}_t^a > 0$ and $\sum_{a \in \hat{\mathcal{G}}_t} \hat{f}_t^a = 1$. The estimated average number of nodes with age-gain $a$ at the start of frame $t$, i.e. $\hat{\Exp}[n_t^a]$, is represented by $\hat{n}_t^a = N \hat{f}_t^a$ as the nodes are symmetric. Since it is impossible to determine the precise value of~the largest age-gain in practical T-DFSA, the AP leverages a threshold $\Gamma_t$ based on estimation $\hat{f}_t^a$, and then allows all nodes with age-gains equal to or higher than $\Gamma_t$ to transmit in frame $t$. The AP then updates $\hat{f}_t^a$ for the subsequent frame based on the observations made during the current frame and its local information. A succinct explanation of the key steps of the proposed practical T-DFSA~(detailed in Section~\ref{sec4}) is given below:
\begin{itemize}
	\item \textbf{Step 1:} The AP decides on $\Gamma_t$ and the frame length $w_t$ based on the estimates $\hat{f}_t^a$, and broadcasts them to the nodes.
	\item \textbf{Step 2:} At the end of frame $t$, the AP improves $\hat{f}_t^a$ using observations made within the frame, which encompass the slot states and age gains of successful nodes. Subsequently, the AP computes the estimated probability of a node with an age-gain of $a$ at the end of frame~$t$, denoted as $\hat{f}_{t^+}^a$.
	\item  \textbf{Step 3:} The AP exploits the update generation rate $\lambda$ and $\hat{f}_{t^+}^a$ to obtain  $\hat{f}_{t+1}^a$.
	\item \textbf{Step 4:} Finally, the AP truncates $\hat{f}_{t+1}^a$ using its knowledge of the maximum age-gain that can be~deduced from the AoI of all nodes, i.e. $y_t^i, \forall i$. 
\end{itemize}\vspace{-0.4em}
\begin{algorithm}[!t]
	\caption{Practical T-DFSA (Implemented in Frame $t$)}
	{\fontsize{9.5}{10.0}\selectfont
	\begin{algorithmic}[1]
		\STATE \textit{// Implemented at the onset of frame $t$}\vskip 1pt
		\renewcommand{\algorithmicrequire}{\textbf{Step 1:~~\emph{Input:}} $\bigl\{\hat{\mathcal{G}}_t, \{\hat{f}_t^a\}_{a \in \hat{\mathcal{G}}_t}\bigr\}$;~~ \textbf{\emph{Output:}} $\{\Gamma_t, w_t\}$}
		\REQUIRE
		\STATE \vskip 1pt Set $\hat{n}_t^a = N \hat{f}_t^a$, $\forall a \in \hat{\mathcal{G}}_t$.\vskip 1pt
		\STATE Derive $\Gamma_t$ and $w_t$ from \eqref{eq10} and \eqref{eq12}, respectively. \vskip 6pt

		\STATE \textit{// Implemented at the end of frame $t$.}\vskip 1pt
		\renewcommand{\algorithmicrequire}{\textbf{Step 2:~~\emph{Input:}} $\bigl\{\hat{\mathcal{G}}_t, \{\hat{n}_t^a\}_{a \in \hat{\mathcal{G}}_t\!}, \mathcal{S}_t, \{n_{t,s}^a\}_{a \in \mathcal{S}_t\!}, (n_S,n_E,n_C)\!\bigr\}$;}
		\REQUIRE \textbf{\emph{Output:}} $\{\hat{f}_{t^+\!}^a\}$
		\vskip 1pt\FOR {$l$ from $(n_S + 2n_C)$ to $N$}
			\STATE $z_0 = 0$\vskip 1pt 
			\IF {$n_S > 0$} \vskip 1pt 
				\STATE Derive $\hat{m}_t^a$, $\forall a \in \mathcal{S}_t$, from Proposition~\ref{prop1}. \vskip 1pt 
				\STATE $z_1 = \dfrac{\binom{w_t}{n_S}\binom{w_t - n_S}{n_E}}{\binom{l}{n_S}} \prod_{a \in \mathcal{S}_t} \binom{\hat{m}_t^a}{n_{t,s}^a} q_S^{n_S} q_E^{n_E} q_C^{n_C}$ \vskip 1pt 
			\ELSE \vskip 1pt 
				\STATE $z_1 = \binom{w_t}{n_E} q_E^{n_E} q_C^{n_C}$ \vskip 1pt 
			\ENDIF \vskip 1pt 
			\IF {$z_1 > z_0$} \vskip 1pt 
				\STATE $z_0 = z_1$ \vskip 1pt 
			\ELSE \vskip 1pt 
				\STATE $\hat{l} = l$ \vskip 1pt 
				\IF {$n_S = 0$} \vskip 1pt 
					\STATE Set $\hat{m}_t^{\Gamma_t} = \hat{l}$ and $\hat{m}_t^a = 0$, $\forall a \in \mathcal{S}_t$. \vskip 1pt 
				\ELSIF {$n_S > 0$} \vskip 1pt 
					\STATE Use $\hat{l}$ in Proposition~\ref{prop1} to derive $\hat{m}_t^a$, $\forall a \in \mathcal{S}_t$. \vskip 1pt 
				\ENDIF \vskip 1pt 
				\STATE Derive $\hat{m}_{t^+}^a$ and $\hat{f}_{t^+}^a$ from \eqref{eq30} and \eqref{eq31}, respectively. \vskip 1pt 
				\STATE \Break \vskip 6pt 
			\ENDIF
		\ENDFOR
		\STATE \textit{// Implemented at the end of frame $t$.}\vskip 1pt
		\renewcommand{\algorithmicrequire}{\textbf{Step 3:~~\emph{Input:}} $\{\hat{f}_{t^+}^a, \max_i x_0^i, \max_i y_t^i\}$;~~\textbf{\emph{Output:}} $\{\hat{f}_{t+1}^a\}$}
		\REQUIRE 
		\STATE \vskip 1pt Derive $\hat{f}_{t+1}^a$ from \eqref{eq34}. \vskip 6pt

		\STATE \textit{// Implemented at the end of frame $t$.}\vskip 1pt
		\renewcommand{\algorithmicrequire}{\textbf{Step 4:~~\emph{Input:}} $\{\hat{f}_{t+1}^a\}$;~~\textbf{\emph{Output:}} $\{\hat{f}_{t+1}^a\}$}
		\REQUIRE 
		\STATE \vskip 1pt Truncate $\hat{f}_{t+1}^a$ using \eqref{eq35}.
	\end{algorithmic}
	}
	\label{alg1}
\end{algorithm}

\section{Characterization of Practical T-DFSA}
\label{sec4}
\fontdimen2\font=0.54ex
This section is dedicated to a comprehensive explanation of the four main steps in practical T-DFSA as given in Algorithm~\ref{alg1}.\vspace{-0.5em}

\subsection{Step 1: Derivation of Threshold and Frame Length}
\label{sec4.1}
\fontdimen2\font=0.54ex
Unlike in ideal T-DFSA, the AP only has access to the~estimates $\hat{n}_t^a$, $\forall a$, in practical T-DFSA.~
Hence, instead of determining the exact highest age-gain, a threshold $\Gamma_t \in \hat{\mathcal{G}}_t$ is determined, indicating that any node with an age-gain equal to or higher than this threshold can transmit in frame~$t$. In~light~of Lemma~\ref{lemma1}, the threshold is chosen to be as high as possible to ensure that only nodes with the highest age-gains contend in frame $t$. Specifically, the threshold is selected to be the highest value that guarantees the estimated average number of active nodes, denoted by $\hat{n}_t$, is at least equal to~one. Owing to this fact, the threshold $\Gamma_t \in \hat{\mathcal{G}}_t$ is chosen such that:
\begin{equation}
	\hat{n}_t = \sum_{a \in \hat{\mathcal{G}}_t: a \geq \Gamma_t}\!\!\! \hat{n}_t^a \geq 1.
	\label{eq9}
\end{equation}
Selecting $\hat{n}_t$ to be less than one would result in empty slots, which is not desirable since the minimum frame length is equal to one. As such, the threshold $\Gamma_t$ can be expressed as:\vspace{-0.2em}
\begin{equation}
	\Gamma_t = \max \left\lbrace \Gamma' \in \hat{\mathcal{G}}_t \mid \sum_{a \in \hat{\mathcal{G}}_t: a \geq \Gamma'} \hat{n}_t^a \geq 1 \right\rbrace.
	\label{eq10}
    \vspace{-0.2em}
\end{equation}
Letting $\hat{n}_t \geq 1$, however, may not be optimal given that decisions are made based on estimates~$\hat{n}_t^a$, rather than the exact values. As a result, we determine $\Gamma_t$ such that $\hat{n}_t$ is lower bounded by a minimal threshold $w_{min}$ as follows:\vspace{-0.3em}
\begin{equation}
	\Gamma_t = \max \left\lbrace \Gamma' \in \hat{\mathcal{G}}_t \mid \sum_{a \in \hat{\mathcal{G}}_t: a \geq \Gamma'} \hat{n}_t^a \geq w_{min} \right\rbrace.
	\label{eq11}
    \vspace{-0.3em}
\end{equation}
Consequently, the frame length is chosen to be:\vspace{-0.2em}
\begin{equation}
	w_t = \myceil{\sum_{a \in \mathcal{A}_t} \hat{n}_t^a},
	\label{eq12}
    \vspace{-0.2em}
\end{equation}
where $\mathcal{A}_t$ is the set of distinct age-gains of active nodes~at~the start of frame $t$. The optimal value~of $w_{min}$ is found~through exhaustive search, but as will be demonstrated in Section~\ref{sec6},~at higher generation rates and lower values of $N$, $w_{min} \!=\! 1$ results~in the best T-DFSA performance because the variance~of the age-gains is reduced and the estimates of $\hat{n}_t^a$ become more accurate. Additionally, the power consumption of the AP (or equivalently, the frequency of its broadcasts) may be restricted.~To guarantee that $w_t$ is always greater than the minimum broadcast interval in this situation, $w_{min}$ can be set to be equal to the inverse of the maximum broadcast frequency.\vspace{-0.7em}

\subsection{Step 2: Derivation of $\hat{f}_{t^+}^a$ using Observations}
\label{sec4.2}
\fontdimen2\font=0.54ex
At the end of frame $t$, the AP updates the list of successfully received age-gains and the status of~the slots within the frame.~It then uses this information to update the available age-gain~distribution $\hat{f}_t^a$ and compute the age-gain distribution at the end of frame $t$, denoted as $\hat{f}_{t^+}^a$, where $t^+$ represents~the end of frame~$t$. The number of singleton, empty, and collision slots, denoted by $n_S$, $n_E$, and $n_C$,~respectively, as well as the age-gains of the successful nodes, are exploited in particular by the AP. As such, the number of successful nodes in frame $t$ with age-gain of $a$ is denoted by $n_{t,s}^a$.

For the sake of better representation, we adopt the notion of \emph{multisets}, which differs from a set in that it allows for~multiple occurrences of an element. A multiset can be denoted as $\{(e, m(e)): e \in A\}$, where $A$ is the set of distinct elements~in the multiset, and $m(e)$ denotes the number of occurrences~of element $e$. Following this definition, $\mathcal{L}_t$ and $\mathcal{L}'_t$, which represent  the multisets of age-gains of active and successful nodes in frame $t$, respectively, are given as:\vspace{-0.2em}
\begin{equation}
   \begin{cases}
    \mathcal{L}_t = \{(a, n_t^a): a \in \mathcal{A}_t\}, \\[2pt]
    \mathcal{L}'_t = \{(a, n_{t,s}^a): a \in \mathcal{S}_t\},
  \end{cases}
  \label{eq13}
  \vspace{-0.2em}
\end{equation}
where $\mathcal{S}_t$ denotes the set of distinct age-gains of successful nodes in frame $t$. Note that all nodes that contend in frame $t$ have age-gains greater than or equal to $\Gamma_t$ and thus, every element in sets $\mathcal{A}_t$ and $\mathcal{S}_t$ is guaranteed to be greater than or equal to $\Gamma_t$.

Let us first suppose that $\mathcal{L}'_t$ is not empty, i.e. there exists at least one successful slot in frame~$t$ (or $n_S > 0$). The~AP then calculates the estimate of $\mathcal{L}_t$, denoted by $\hat{\mathcal{L}}_t$, based on the observations $I=(n_S,n_E,n_C)$ and $\mathcal{L}'_t$ as follows:\vspace{-0.1em}
\begin{equation}
	\hat{\mathcal{L}}_t = \bigl\{(a, \hat{m}_t^a): a \in \hat{\mathcal{A}}_t\bigr\},
	\label{eq14}
	\vspace{-0.2em}
\end{equation}
where $\hat{\mathcal{A}}_t$ and $\hat{m}_t^a$ are estimations of $\mathcal{A}_t$ and average $n_t^a$ after observing frame $t$, respectively.~Specifically, we can conveniently write:\vspace{-0.2em}
\begin{equation}
	\hat{m}_t^a = \hat{\Exp}\left[n_t^a | I, \mathcal{L}'_t\right].
	\label{eq15}
	\vspace{-0.2em}
\end{equation}
We now employ the maximum likelihood (ML) estimation to derive $\hat{\mathcal{L}}_t$ as follows:\vspace{-0.2em}
\begin{equation}
	\hat{\mathcal{L}}_t = \argmax_{\mathcal{L}_t} \Pr(I, \mathcal{L}'_t | \mathcal{L}_t),
	\label{eq16}
	\vspace{-0.2em}
\end{equation}
where the maximization in \eqref{eq16} can be further expressed as:\vspace{-0.2em}
\begin{equation}
	\max_{l} \max_{\mathcal{L}_t: |\mathcal{L}_t|=l} \Pr(I | \mathcal{L}_t) \Pr(\mathcal{L}'_t | \mathcal{L}_t,I).
	\label{eq17}
 \vspace{-0.2em}
\end{equation}

Next, we focus on deriving the first conditional probability in \eqref{eq17}, i.e. $\Pr(I | \mathcal{L}_t)$. The probabilities that a single time slot in frame $t$ is a singleton, an empty, and a collision slot are denoted by $q_S$, $q_E$, and $q_C$, respectively, and can be readily computed as:
\begin{equation}
   \begin{cases}
    q_S = \frac{|\mathcal{L}_t|}{w_t}\left(1 - \frac{1}{w_t}\right)^{|\mathcal{L}_t|-1}, \vspace{0.2em} \\
    q_E = \left(1 - \frac{1}{w_t}\right)^{|\mathcal{L}_t|}, \vspace{-0.2em} \\
    q_C = 1 - q_S - q_E,
  \end{cases}
  \label{eq18}
\end{equation}
where $|\mathcal{L}_t| = \sum_{a \geq \Gamma_t} n_t^a$. It should be noted that $q_S$ in \eqref{eq18} is the likelihood that one of the $|\mathcal{L}_t|$ nodes transmits in a chosen slot, with probability $1/w_t$, and the others do not, with probability $(1 - 1/w_t)^{|\mathcal{L}_t|-1}$, whereas $q_E$ is the probability that no node transmits. Using \eqref{eq18}, $\Pr(I | \mathcal{L}_t)$ can now be stated as follows:\vspace{-0.2em}
\begin{equation}
	\Pr(I | \mathcal{L}_t) = \Pr(I |\, |\mathcal{L}_t|) = \binom{w_t}{n_S}\binom{w_t-n_S}{n_E} q_S^{n_S} q_E^{n_E} q_C^{n_C}.
	\label{eq19}
\end{equation}
Interestingly, $\Pr(I | \mathcal{L}_t)$ is solely dependent on the number of active nodes and not their age-gains. Thereby, the maximization in \eqref{eq17} can be re-written as:\vspace{-0.2em}
\begin{equation}
	\max_{l} \Pr(I |\, |\mathcal{L}_t| = l)\,\,\max_{\mathcal{L}_t: |\mathcal{L}_t|=l} \Pr(\mathcal{L}'_t| \mathcal{L}_t, I).
	\label{eq20}
    \vspace{-0.2em}
\end{equation}

Let $\hat{\mathcal{L}}_t(l)$ be the optimal estimation of the multiset $\mathcal{L}_t$, given $\mathcal{L}_t$, $I$, and that the number of active nodes, $|\mathcal{L}_t|$, is equal to $l$. That is to say,\vspace{-0.2em}
\begin{equation}
	\hat{\mathcal{L}}_t(l) = \argmax_{\mathcal{L}_t: |\mathcal{L}_t|=l} \Pr(\mathcal{L}'_t | \mathcal{L}_t, I).
	\label{eq21}
	\vspace{-0.2em}
\end{equation}
Resulting from \eqref{eq20} and \eqref{eq21}, the optimal estimation of the number of active nodes, denoted by $\hat{l}$, can be obtained as:
\begin{equation}
	\hat{l} = \argmax_{l} \Pr(I |\, |\mathcal{L}_t| = l) \Pr(\mathcal{L}'_t | \hat{\mathcal{L}}_t(l), I).
	\label{eq22}
\end{equation}
Moreover, $\hat{\mathcal{L}}_t = \hat{\mathcal{L}}_t(\hat{l})$. With a slight abuse of notation, we use $\hat{\mathcal{L}}_t(l) = \{(a, \hat{m}_t^a): a \in \hat{\mathcal{A}}_t | \sum_a \hat{m}_t^a = l\}$ to derive $\hat{\mathcal{L}}_t(l)$ in Proposition~\ref{prop1}.
\begin{algorithm}[!t]
	\caption{\,\,Derivation of $r_a$ (introduced in Proposition~\ref{prop1})}
	{\fontsize{9.5}{10.0}\selectfont
	\begin{algorithmic}[1]
		\vskip 1pt\renewcommand{\algorithmicrequire}{\textbf{\emph{Input:}} $\{l, n_S, \{n_{t,s}^a\}_{a \in \mathcal{S}_t}\}$;\,\, \textbf{\emph{Output:}} $\{r_a, \forall a \in \mathcal{S}_t\}$}
 		\REQUIRE 
 		\STATE \vskip 2pt Set $(\hat{m}_t^a, r_a) = \left(\lfloor \frac{l}{n_S}\rfloor n_{t,s}^a, 0\right), \forall a \in \mathcal{S}_t$.\vskip 1pt
 		\STATE Set $r = l - \lfloor \frac{l}{n_S}\rfloor n_S$.\vskip 1pt
 		\vskip 1pt\WHILE {$r > 0$}
 			\STATE $\hat{m}_t^J = \hat{m}_t^J + 1$  \qquad\quad $\triangleright$ where $J = \argmin_{a\in \mathcal{S}_t} \dfrac{\hat{m}_t^a + 1}{n_{t,s}^a}$
 			\STATE $r_J = r_J + 1$\vskip 1pt
 			\STATE $r = r - 1$
 		\ENDWHILE
 \end{algorithmic}
 }
 \label{alg2}
\end{algorithm}%
\begin{prop}
	\label{prop1}
	$\hat{\mathcal{L}}_t(l)$ is calculated by setting $\hat{\mathcal{A}}_t = \mathcal{S}_t$ and $\hat{m}_t^a = \lfloor l/n_S \rfloor n_{t,s}^a + r_a$, where $0 \leq r_a \leq n_{t,s}^a$ is derived~from Algorithm~\ref{alg2}. Moreover, the probability $\Pr(\mathcal{L}'_t | \hat{\mathcal{L}}_t(l), I)$ in \eqref{eq22} is computed~as:\vspace{-0.2em}
	\begin{equation}
	\Pr(\mathcal{L}'_t | \hat{\mathcal{L}}_t(l), I) = \frac{\prod_{a \in \mathcal{S}_t} \binom{\hat{m}_t^a}{n_{t,s}^a}}{\binom{l}{n_S}}.
	\label{eq23}
	\vspace{-0.2em}
\end{equation}
\end{prop}
\begin{proof}
	Knowing $I$ in \eqref{eq21} implies that $|\mathcal{L}'_t| = n_S$. Furthermore, given that $|\mathcal{L}_t| = l$ and $|\mathcal{L}'_t| = n_S$, the probability of observing $\mathcal{L}'_t$ depends on which specific sets of $n_S$ nodes out of $l$ active nodes transmit successfully in frame $t$. In light of the fact that all active nodes behave symmetrically for contention within a frame, the probability that a particular set of $n_S$ nodes would succeed is equal to $1/\binom{l}{n_S}$ and is therefore, independent of $n_E$ and $n_C$. Consequently, the maximization in \eqref{eq21} can be expressed as:\vspace{-0.2em}
	\begin{equation}
		\hat{\mathcal{L}}_t(l) = \argmax_{\mathcal{L}_t} \Pr(\mathcal{L}'_t \mid \mathcal{L}_t, |\mathcal{L}_t| = l, |\mathcal{L}'_t| = n_S),
		\label{eq43}
        \vspace{-0.2em}
	\end{equation}
	where the probability in \eqref{eq43} can be obtained as follows:\vspace{-0.2em}
	\begin{equation}
		\\Pr(\mathcal{L}'_t \mid \mathcal{L}_t, |\mathcal{L}_t| = l, |\mathcal{L}'_t| = n_S) = \frac{\prod_{a \in \mathcal{S}_t} \binom{n_t^a}{n_{t,s}^a}}{\binom{l}{n_S}}.
		\label{eq44}
	\end{equation}
	The denominator in \eqref{eq44} is the total number of (not necessarily distinct) multisets of size $n_S$ observable in frame $t$. The numerator in \eqref{eq44} represents the number of ways the multiset $\mathcal{L}'_t$ can be observed, i.e., $\forall a \in \mathcal{S}_t$, the combination of $n_{t,s}^a$ nodes with age-gain $a$ out of $n_t^a$ active nodes with the same age-gain. Replacing $\mathcal{L}_t$ in \eqref{eq44} with $\hat{\mathcal{L}}_t(l)$ results in \eqref{eq23}. We derive $\hat{m}_t^a$ in the rest of the proof. As such, we first derive a necessary condition for $\hat{m}_t^a$ in Lemma~\ref{lemma2} using \eqref{eq23}. Lemma~\ref{lemma3} then proceeds with the derivation of the exact values of $\hat{m}_t^a$.%
	\begin{lem}
		\label{lemma2}
		For $\hat{\mathcal{L}}_t(l)$, we have $\hat{\mathcal{A}}_t = \mathcal{S}_t$. Moreover, $\forall a,b \in \hat{\mathcal{A}}_t$, the following condition holds:
		\begin{equation}
			\frac{\hat{m}_t^b}{n_{t,s}^b} - \frac{\hat{m}_t^a}{n_{t,s}^a} \leq \frac{1}{n_{t,s}^a}.
			\label{eq24}
		\end{equation}
	\end{lem}
	\begin{proof}
		See Appendix~\ref{appB}.
	\end{proof}
	\begin{lem}
		\label{lemma3}
		For any $a \in \hat{\mathcal{A}}_t$, we have the following, where $0 \leq r_a \leq n_{t,s}^a$, $k' = \lfloor l/n_S\rfloor$, and $k', r_a \in \mathbb{Z}^+$:
		\begin{equation}
			\frac{\hat{m}_t^a}{n_{t,s}^a} = k' + \frac{r_a}{n_{t,s}^a}.
			\label{eq25}
		\end{equation}
	\end{lem}
	\begin{proof}
		See Appendix~\ref{appC}.
	\end{proof}
	
In essence, Lemma~\ref{lemma3} postulates that $\hat{m}_t^a / n_{t,s}^a$ is either $k'$ or $k'+1$ for all values of $a$. The ultimate step is to determine the $r_a$ values in Lemma~\ref{lemma3}. To this end, we present Algorithm~\ref{alg2} which sequentially increases the values of $\hat{m}_t^a$ while ensuring that  \eqref{eq24} holds in each iteration, thus preserving the viability of the final $\hat{m}_t^a$ values. In this algorithm, we apply Lemma~\ref{lemma3} to initialize $\hat{m}_t^a$ with $\lfloor l / n_S\rfloor n_{t,s}^a$, which satisfies \eqref{eq24}. Using induction, we then assume that \eqref{eq24} holds in the $i$-th iteration, $\forall a,b \in \hat{\mathcal{A}}_t$. Then, in iteration $i+1$, we set:
\begin{equation}
  \hat{m}_t^a(i+1) =
  \begin{cases}
    \hat{m}_t^a(i) + 1, & \text{ if $a = J = \argmin_{a \in \hat{\mathcal{A}}_t\!} \frac{\hat{m}_t^a(i) + 1}{n_{t,s}^a}$} \\[2pt]
    \hat{m}_t^a(i), & \text{ if $a \in \hat{\mathcal{A}}_t$ and $a \neq J$}.
  \end{cases}
  \label{eq26}
\end{equation}%
Since, $\forall a,b \neq J$, $\hat{m}_t^a(i + 1)$ and $\hat{m}_t^b(i + 1)$ retain the same values as in the previous iteration, the corresponding inequality \eqref{eq24} still holds in iteration $i + 1$. We now need to show that it also holds for $J$ and $\forall a \neq J$. More precisely, we should prove that $\hat{m}_t^a(i+1) / n_{t,s}^a \leq (1 + \hat{m}_t^J(i+1)) / n_{t,s}^J$ and $\hat{m}_t^J(i+1) / n_{t,s}^J \leq (1 + \hat{m}_t^a(i+1)) / n_{t,s}^a$. Using the formulation in \eqref{eq26} allows us to further simplify the above two inequalities as $\hat{m}_t^a(i) / n_{t,s}^a \leq (2 + \hat{m}_t^J(i)) / n_{t,s}^J$ and $(\hat{m}_t^J(i) + 1) / n_{t,s}^J \leq (1 + \hat{m}_t^a(i)) / n_{t,s}^a$, respectively. The first inequality holds directly from the induction assumption, whereas the second inequality is asserted by the definition of $J$. This completes the proof.
\end{proof}


By plugging \eqref{eq19} and \eqref{eq23} in \eqref{eq22}, we obtain the following, where $q_S$, $q_E$, and $q_C$ are derived from \eqref{eq18} by setting $|\mathcal{L}_t| = l$:
\begin{equation}
	\hat{l} = \argmax_l \frac{\binom{w_t}{n_S} \binom{w_t - n_S}{n_E}}{\binom{l}{n_S}} \prod_{a \in \mathcal{S}_t} \binom{\hat{m}_t^a}{n_{t,s}^a} q_S^{n_S} q_E^{n_E} q_C^{n_C}.
	\label{eq27}
\end{equation}
Besides, $\hat{m}_t^a$ is a function of $l$ and can be readily obtained from Proposition~\ref{prop1}. Once $\hat{l}$ is computed using \eqref{eq27}, substituting it in Proposition~\ref{prop1} eventually yields \eqref{eq14}.

In our derivation of $\hat{\mathcal{L}}_t$ so far, $\mathcal{L}'_t$ was assumed to be non-empty. For the case when $\mathcal{L}'_t = \emptyset$, the optimization in \eqref{eq20} reduces to:\vspace{-0.2em}
\begin{equation}
	\hat{\mathcal{L}}_t = \argmax_{\mathcal{L}_t} \Pr(I \mid \mathcal{L}_t),
	\label{eq28}
    \vspace{-0.2em}
\end{equation}
where $\Pr(I|\mathcal{L}_t)$ is derived from \eqref{eq19} and is only dependent on $|\mathcal{L}_t|$. Thus, we get:\vspace{-0.2em}
\begin{equation}
	\hat{l} = \argmax_{l} \binom{w_t}{n_E} q_E^{n_E} q_C^{n_C},
	\label{eq29}
    \vspace{-0.2em}
\end{equation}
where $q_E$ and $q_C$ are deduced from \eqref{eq18} by setting $n_S = 0$ and $|\mathcal{L}_t| = l$. As can be observed, we are only able to estimate the number of active nodes when $n_S = 0$ and not their age-gains. We therefore, suppose that all $\hat{l}$ active users have an age-gain of $\Gamma_t$ in this circumstance.

Thus far, we have updated the estimation $\hat{n}_t^a$ for $a \geq \Gamma_t$ using the observations in frame $t$. Our ultimate goal is to calculate $\hat{f}_{t^+}^a$, or the probability that a generic node has an age-gain of $a$ at the end of the frame $t$. To achieve this, we first deduce $\hat{m}_{t^+}^a, \forall a$, using the earlier estimations $\hat{n}_t^a$ for $a < \Gamma_t$, the derived estimations $\hat{m}_t^a$ for $a \geq \Gamma_t$, and the data on $n_{t,s}^a$ as follows:\vspace{-0.2em}
\begin{equation}
  \hat{m}_{t^+}^a =
  \begin{cases}
    \hat{n}_t^0 + n_S, & \text{ if $a = 0$} \\[0.5pt]
    \hat{n}_t^a,       & \text{ if $a \in \hat{\mathcal{G}}_t, 0 < a < \Gamma_t$} \\[0.5pt]
     \hat{m}_t^a - n_{t,s}^a, & \text{ if $a \in \mathcal{S}_t$ } \\[0.5pt]
    0, &  \text{ if $a \in \hat{\mathcal{G}}_t$,  $a\geq \Gamma_t$, $a \not\in \mathcal{S}_t$}\,.
  \end{cases}
  \label{eq30}
  \vspace{-0.2em}
\end{equation}
In the first case of \eqref{eq30}, $n_S$ number of successful nodes with zero age-gain at the end of frame $t$ are added to the initial estimation of $\hat{n}_t^0$. In the second case, $\hat{n}_t^a$ for $a < \Gamma_t$ remains unaffected since the corresponding nodes are inactive in frame $t$. For $a \in \mathcal{S}_t$ in the third case of \eqref{eq30}, $n_{t,s}^a$ is deducted~from the updated estimation of $n_t^a$, i.e. $\hat{m}_t^a$. In the final case, if $a \geq \Gamma_t$ but is not among the observed age-gains, we set $\hat{m}^a_{t^+}=0$. This case is considered because, according to Proposition~\ref{prop1}, the predicted active age-gains in frame $t$ match the observed ones, i.e. $\mathcal{A}_t=\mathcal{S}_t$. Therefore, if $a \in \hat{\mathcal{G}}_t$ and $a \geq \Gamma_t$, but it is not amongst the observed age-gains, then the estimation $\hat{m}^a_t$, which is the updated value~of $\hat{n}^a_t$ after observations, is set to zero. As a result, the corresponding estimated value at the end of~frame $t$, $\hat{m}^a_{t^+}$, is also set to zero, as stated in the last case of \eqref{eq30}. It is worth noting that in the third and~fourth cases of \eqref{eq30}, which correspond to $a \geq \Gamma_t$ (either $a \in \mathcal{S}_t$ or not), the values of $\hat{m}_{t^+}^a$ are solely derived based on observations and not on the previous estimations $\hat{n}^a_t$. Conversely, in the case where $a < \Gamma_t$, the values $\hat{m}_{t^+}^a$ are set to be equal to the previous estimations $\hat{n}^a_t$. Accordingly,~although~the $\hat{n}_{t}^a$ values satisfy the equality $\sum_a \hat{n}_{t}^a \!=\! \sum_a N\hat{f}_{t}^a \!=\! N$, the constraint $\sum_a \hat{m}_{t^+}^a \!=\! N$ does not necessarily hold.~In this respect, we normalize the $\hat{m}_{t^+}^a$ values by dividing each value by their sum. This normalization allows us to derive the probability distribution $\hat{f}_{t^+}^a$ as follows:\vspace{-0.2em}
\begin{equation}
	\hat{f}_{t^+}^a = \frac{\hat{m}_{t^+}^a}{\sum_{a \in \hat{\mathcal{G}}_{t^+}} \hat{m}_{t^+}^a},
	\label{eq31}
    \vspace{-0.2em}
\end{equation}
where $\hat{G}_{t^+}$ denotes the set of existing age-gains at the end of frame $t$, formulated as:\vspace{-0.2em}
\begin{equation}
	\hat{\mathcal{G}}_{t^+} = \{a | a \in \hat{\mathcal{G}}_t, a < \Gamma_t\} \cup \{a | a \in \mathcal{S}_t, \hat{m}_{t^+}^a > 0\}.
	\label{eq32}
    \vspace{-0.2em}
\end{equation}

In Algorithm~\ref{alg1}, lines 5 to 21 implement Step 2 of practical T-DFSA. This includes solving the optimization in \eqref{eq27} using exhaustive search. The minimum number of active nodes in frame $t$ is set to be $n_S + 2 n_C$ since there must be at least two contesting active users in each collision slot (line~5). Note that the incurred time complexity is minimal due to the generally short frame lengths.\vspace{-0.6em}

\subsection{Step 3: Derivation of $\hat{f}_{t+1}^a$ using Generation Rates}
\label{sec4.3}
\fontdimen2\font=0.54ex
Estimates of $\hat{f}_{t^+}^a$ quantify the impact of observations in frame $t$ but not the generation rates of~status updates in this frame. In the proposition below, we exploit $\hat{f}_{t^+}^a$ and the update rates $\lambda$ to obtain the estimations for $\hat{f}_{t+1}^a$.\vspace{-0.2em}
\begin{prop}
\label{prop2}
	Using $\hat{f}_{t^+}^a$ derived in \eqref{eq31}, the values of $\hat{f}_{t+1}^a$ can be estimated as follows:\vspace{-0.3em}
	\begin{equation}
	\hat{f}_{t+1}^a = (1 - \lambda)^{w_t} \hat{f}_{t^+}^a + \sum_{b=0}^{a-1} \hat{f}_{t^+}^b \sum_{c = c_{min}}^{c_{max}} p_c\, \tilde{p}_h,
	\label{eq33}
	\vspace{-0.2em}
\end{equation}
where $p_c = \lambda(1 - \lambda)^{c-1}$ and $h = c + a - b - w_t$. Here, $\tilde{p}_h$ is the probability that the age of the node at the start of frame~$t$ is $h$ and is given by:\vspace{-0.2em}
	\begin{equation}
  \tilde{p}_h =
  \begin{cases}
    \dfrac{\lambda (1 - \lambda)^{h-1}}{1 - (1 - \lambda)^{h_{max}}}, & \text{if $1 \leq h \leq h_{max}$} \\[1pt]
    0,       & \text{if otherwise},
  \end{cases}
  \label{eq34}
  \vspace{-0.2em}
\end{equation}
where $h_{max} = \max_i x_t^i = \min \{\max_i x_0^i + k_t, \max_i y_t^i - b\}$, $c_{min} = (b - a + w_t)^+ + 1$, $c_{max} = \max \{(h_{max} - a + b)^- + w_t, 1\}$, $(z)^+ = \max \{0, z\}$, and $(z)^- = \min \{0,z\}$.
\end{prop}
\begin{proof}
	See Appendix~\ref{appE}.
\end{proof}
\begin{figure*}[!t]
	\centering
	\includegraphics[width=0.9\textwidth]{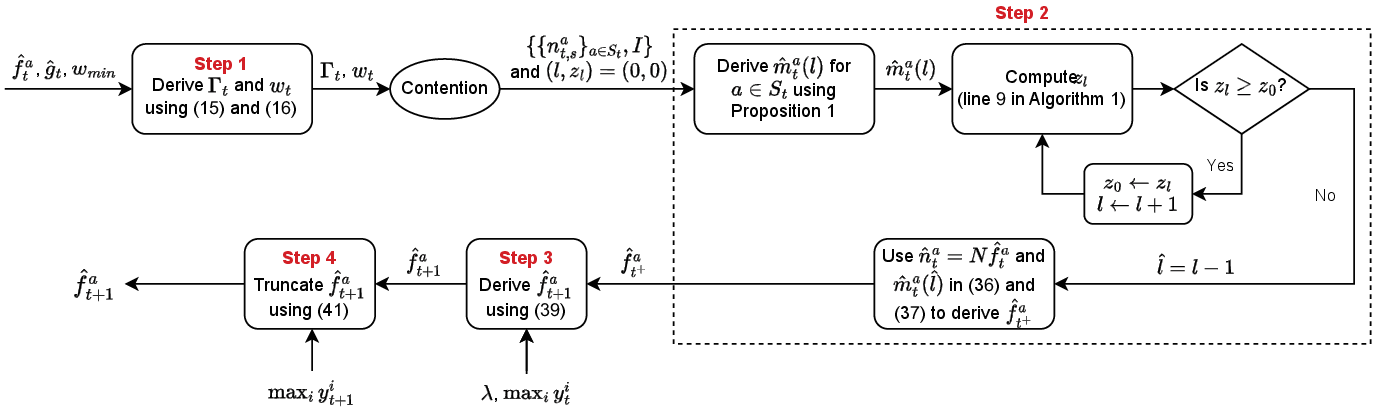}
	\vspace{-0.7em}
	\caption{Flowchart of T-DFSA illustrating the four main steps given in Algorithm~\ref{alg1}.}
	\label{fig:flowchart}
	\vspace{-0.3em}
\end{figure*}

In Proposition~\ref{prop2}, $p_c$ essentially captures the probability of a fresh update being generated at a~typical node in time slot $k_{t+1} - c$ within frame $t$. Moreover, the age of a typical node at the onset of frame~$t$, i.e. $h$, establishes a relation between the age-gain at the onset of that frame ($a$), the age-gain just before the end of the frame ($b$), $w_t$, and $c$ (more details in Appendix~\ref{appE}). In the derived upper bound for $h$, $\max_i x_0^i + k_t$ refers to the maximum possible value of $h$ when no updates are generated~by any of the $N$ sources up to time slot $k_t$\footnote{In T-DFSA, the nodes only retain fresh updates at time $0$, i.e. $x_0^i=1$. As a result, the maximum value $\max_i x_0^i$~is known at the AP, and there is no need for the AP to acquire $\{x_0^i\}_i$. This assumption is specifically related to the transient phase~of the network and does not impact the steady-state AAoI of the network.}. On the other hand, we have $g_t^i \geq g_{t^+}^i$, which means that~the age-gain may fall at the end of the frame as a result of a successful transmission. Then, given $g_{t^+}^i = b$ and $g_t^i = y_t^i - x_t^i$, we establish that $x_t^i \leq y_t^i - b$. Hence, the expression $\max_i y_t^i - b$ indicates the highest feasible value of $h$ in relation to the $y_t^ i$ values. Interestingly enough, the former limit dominates in the initial stages of practical T-DFSA, while the latter gains prominence when the protocol is implemented for a sufficiently large time interval.
\begin{rem}
    In the upper bound expression $\max_i y^i_t - b$, we~assume that $b \leq \max_i y^i_t - 1$ to ensure that $h_{\max} \geq 1$. This implies $\hat{f}^b_{t^+} = 0$ for $b \leq \max_i y^i_t - 1$. We can satisfy this condition in Step 4 of Algorithm~\ref{alg1}, where $\hat{f}^a_t$ is truncated at the end of frame $t-1$ to guarantee $\hat{f}^a_t = 0$ for $a > \max_i y^i_t$. Additionally, considering that $g_t^i \geq g_{t^+}^i$, we can conclude that $\hat{f}^b_{t^+} = 0$ for $b > \max_i y^i_t - 1$.
    \label{rem1}
\end{rem}
\vspace{-1.8em}
\begin{rem}
In Proposition~\ref{prop2}, $h_{max}=\max_i y^i_t-b$ is in steady state. Thus, the inequality $h = c + a - b - w_t \leq h_{max}$ implies~that $a \leq \max_i y^i_t + w_t - c \leq \max_i y^i_t + w_t - 1$. Consequently, $\hat{f}_{t+1}^a = 0$~holds for $a > \max_i y^i_t+w_t$.
\label{rem2}
\end{rem}

The above remarks highlight the importance of considering an upper bound for $h$ by utilizing information on $\max_i y^i_t$. This approach allows us to avoid computing $\hat{f}^a_{t+1}$ for large values~of $a$,~which could lead to an overestimation of the threshold $\Gamma_t$. Overestimating the threshold would result in~the inability to transmit backlogged updates, thus leading to their replacement with newer updates.~In turn, more backlogged nodes with higher age-gains would be spawned which may render the network unstable. Note that the network is said to be stable if and only if $\lim_{t \rightarrow \infty} \sum_{i = 1}^N y_t^i / N = 0$. Furthermore, the upper bound on the age-gains in this step draws out from $\max_i y_t^i$ and $\max_i x_0^i$ (line 23 of Algorithm~\ref{alg1}). The subsequent step further refines the range of age-gains in $\hat{f}_{t+1}^a$ using $\max_i y_{t+1}^i$.\vspace{-1em}

\subsection{Step 4: Truncation of $\hat{f}_{t+1}^a$}
\label{sec4.4}
\fontdimen2\font=0.54ex
The last step (line 25 of Algorithm~\ref{alg1}) truncates $\hat{f}_{t+1}^a$ by~reducing the range of potential $a$ values using the data in $y_{t+1}^i$ that is available at the AP. Beyond any doubt, we know that the maximum age-gain at the AP at the start of frame $t+1$ cannot exceed $a_{max} \triangleq \max_i y_{t+1}^i - 1$, where the~maximum value is achieved when $u^i(k_{t+1} - 1) = 1$ for $i = \argmax_i y_{t+1}^i$. Thus, we update $\hat{f}_{t+1}^a$ as follows:\vspace{-0.2em}
\begin{equation}
  \hat{f}_{t+1}^a =
  \begin{cases}
    \dfrac{\hat{f}_{t+1}^a}{\sum_{a = 0}^{a_{max}}\hat{f}_{t+1}^a}, & \text{ if $0 \leq a \leq a_{max}$}  \\[1pt]
    0,       & \text{ if otherwise}.
  \end{cases}
  \label{eq35}
  \vspace{-0.3em}
\end{equation}
The truncation in \eqref{eq35} profoundly impacts the stability~of T-DFSA since it prevents overestimating the threshold. The flowchart of the four steps in T-DFSA is shown in Figure \ref{fig:flowchart}.\vspace{-0.6em}

\section{Discussions on Complexity and Heterogeneity}
\label{sec5}
\fontdimen2\font=0.54ex

\subsection{Complexity Analysis of Stable T-DFSA}
\label{sec5_1}
\fontdimen2\font=0.54ex
As was briefly pointed out earlier, stability is a crucial~requirement in random access protocols.~Our numerical findings show that the T-DFSA protocol is highly likely to be stable if the initial age-gain values are chosen to be sufficiently diverse. Even if the initial age-gain values are not diverse enough,~instability only occurs in the early phases, when congestion may arise. In this case, the AP can detect instability when there are multiple consecutive frames with no successful transmission. The AP~can reset the source nodes by instructing them to discard their~previous updates and retain only the most recent ones, if any, before resuming the protocol.

\begin{rem}
The complexity of the T-DFSA algorithm in frame $t$ can be expressed as $O(n_S N) + O(w_t \max_i y^i_t)$.
\label{rem3}
\end{rem}
The first term of the complexity introduced in Remark~\ref{rem3}~pertains to finding the optimal value of~$l$ (i.e., the number of active nodes) in the loop of Step 2 in Algorithm~\ref{alg1} (lines 5 to 21). Within each iteration of the loop, the main computational burden lies in computing $z_l$ in line 9, while executing Algorithm~\ref{alg2} incurs low complexity. This is because in an ideally stable T-DFSA, the frame lengths are often $w_{min} + 1$ or $w_{min} + 2$, thus resulting in $n_S \leq w_{min} + 1$. Hence, the $\mathbf{while}$~loop in~Algorithm~\ref{alg2}, repeated for $r < n_S$ iterations, has a low time complexity. It should be noted that $O(n_S N)$ is a theoretical bound, and the search for the optimal $l$ does not extend up to $N$ since frame lengths~are typically small, thus restricting the number of active nodes when $n_S > 0$. Only when all slots of a frame are marked as collision slots ($n_S = 0$), will the search proceed up to $N$. However, as shown in Section~\ref{sec6}, this occurrence is exceedingly rare and has an infinitesimal probability in stable T-DFSA scenarios. The second term of the complexity, $O(w_t \max_i y^i_t)$, relates to Step 3 of Algorithm~1, which involves computing $\hat{f}^a_{t+1}$ for all $a \leq \max_i y_t^i + w_t - 1$. This complexity increases as $\lambda$ decreases~and $N$ grows. As $\lambda$ decreases, the inter-arrival times between updates at the nodes increase, leading to larger values of $x_t^i$ and consequently larger $y_t^i$. On the flip side, as $N$ grows, more contentions occur among nodes, causing them to transmit updates with larger delays which inevitably results in larger AoIs at the AP (i.e., larger $\max_i y_t^i$).\vspace{-0.7em}

\subsection{Effect of Heterogeneity of Nodes}
\label{sec5_2}
\fontdimen2\font=0.54ex
In this work, we have primarily focused on the \textit{symmetric} scenario, where all nodes exhibit identical arrival rates and~engage in symmetric transmission for contention. However, our proposed algorithm can be extended to accommodate scenarios with diverse node types. In such heterogeneous settings, each node type $x \in \mathcal{X}$ has a specific arrival rate, $\lambda_x$, and transmission probability, $p_x$. In fact, the only symmetric aspect that needs~to be preserved in our approach is the symmetric behavior of active nodes for contention, where they select one slot randomly within the frame for transmission. This is crucial because the results in Lemma 1 and Proposition 1 rely on this assumption.~In~the heterogeneous case, each node of type $x$ transmits with a probability of $p_x$ when its age-gain exceeds $\Gamma_t$. Since nodes of the same type are symmetric, we estimate the probability $\hat{f}^{a,x}_t$ instead of $\hat{f}^a_t$, which signifies the likelihood that a node of type $x$ has an age-gain of $a$ at the beginning of frame $t$. Then, $\hat{n}^{a,x}_{t,on} \!=\! N_x p_x \hat{f}^{a,x}_t$, where $N_x$ denotes the number of~nodes of type $x$, represents~the average number of active nodes of type $x$ with an age-gain of $a$, given that $a \geq \Gamma_t$. Other relationships~presented in the preceding sections can be similarly extended for the heterogeneous case.\vspace{-0.8em}

\begin{figure}[!t]
  \centering
  \subfloat[AAR versus $w_t$ and $\Gamma_t$.]{
    \includegraphics[width=0.45\columnwidth]{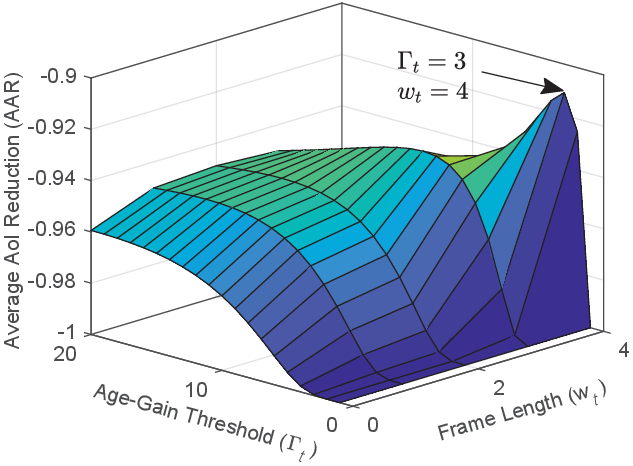}
    \label{fig:verify-lemma1}
  }
  \hfil
  \subfloat[$\Pr(\mathcal{L}'_t|\mathcal{L}_t,I)$ versus $n^2_t$ and $n^5_t$.]{
    \includegraphics[width=0.45\columnwidth]{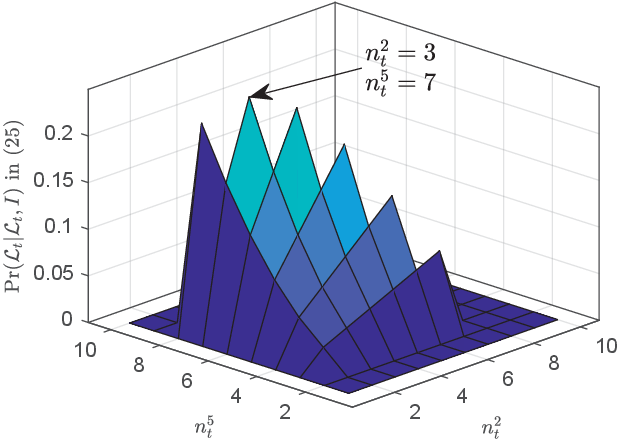}
    \label{fig:verify-prop1}
  }
  \caption{Verification of (a) Lemma~\ref{lemma1} where $N=20$, $M=4$, $n^0_t=4$, $n^1_t=1$, $n^2_t=n^3_t=6$, and $n^4_t=3$,~and (b) Proposition~\ref{prop1} where $\mathcal{S}_t=\{2,5\}$, $n^2_{t,s}=2$, $n^5_{t,s}=1$, $l=10$, and $w_t=10$.}
  \label{fig:mainfigure}
  \vspace{-0.2em}
\end{figure}
\section{Numerical Results and Discussions}
\label{sec6}
\fontdimen2\font=0.54ex
We now evaluate the performance of T-DFSA and benchmark it against the optimal FSA, TA \cite{Tahir2021}, SAT, and AAT~\cite{Chen2022}~protocols in terms of the network AAG, which is defined by the expression $\lim_{T \rightarrow \infty} \sum_{k=0}^T \sum_{i=1}^N (y_k^i - x_k^i)/(NT)$. Here, the first term $\lim_{T \rightarrow \infty} \sum_{k=0}^T \sum_{i=1}^N y_k^i/(NT)$ signifies the AoI of the network, whereas $\lim_{T \rightarrow \infty} \sum_{k=0}^T \sum_{i=1}^N x_k^i/(NT) = 1/\lambda$. The Monte Carlo simulations were run in the MATLAB environment for sufficiently long time intervals.

\figurename{~\ref{fig:verify-lemma1}} and \figurename{~\ref{fig:verify-prop1}} serve as verification for the validity of Lemma~\ref{lemma1} and Proposition~\ref{prop1}, respectively. In \figurename~\ref{fig:verify-lemma1}, we set $N=20$, $M=4$ (i.e., the maximum age-gain), $n^0_t=4$, $n^1_t=1$, $n^2_t=n^3_t=6$,~and $n^4_t=3$\footnote{These values have been deliberately chosen to be small for better visualization, and the curves are plotted as continuous~values.}. By varying the parameters $w_t$ and $\Gamma_t$ and simulating the contention process described~in Section~\ref{sec2}, the figure plots the AAR against $w_t$ and $\Gamma_t$. As observed in \figurename{~\ref{fig:verify-lemma1}}, the maximum~AAR is attained at $w_t=4$ and $\Gamma_t=3$, which~is consistent with the findings of Lemma~\ref{lemma1}. In \figurename{~\ref{fig:verify-prop1}},~another scenario is considered where $\mathcal{S}_t=\{2,5\}$, $n^2_{t,s}=2$, $n^5_{t,s}=1$, $l \!=\! 10$ (i.e., the number of active~nodes), and $w_t \!=\! 10$. In this figure, the probability $\Pr(\mathcal{L}'_t|\mathcal{L}_t,I)$ from \eqref{eq21} is tracked for various combinations of $n^5_t$ and $n^2_t$ such that $n^5_t+n^2_t \leq 10$. As depicted in \figurename{~\ref{fig:verify-lemma1}}, the probability is maximized when $n^2_t=3$ and $n^5_t=7$, thus aligning with the values proposed by Proposition \ref{prop1}. Hence, Proposition~\ref{prop1} suggests setting $n^2_t=\lfloor 10/3 \rfloor n^2_{t,s} +r_2=3\times 2+1$ and $n^2_t=\lfloor 10/3 \rfloor n^5_{t,s} +r_5=3\times 1+ 0$.
\begin{figure}[!t]
	\centering
 \includegraphics[width=0.6\columnwidth]{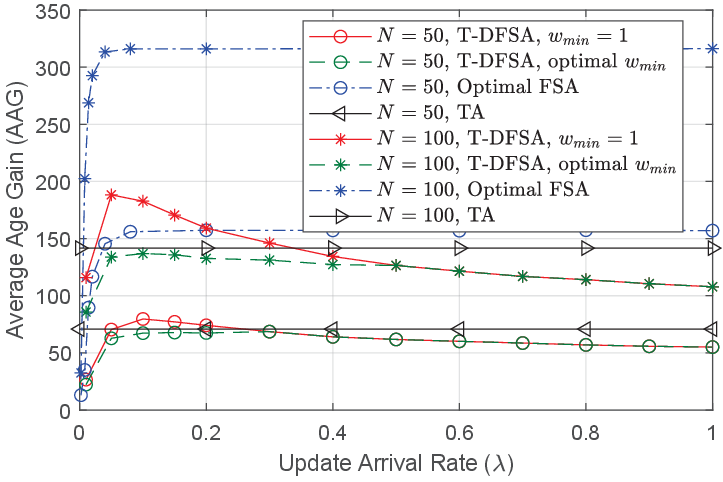}
	\vspace{-0.7em}
	\caption{AAG comparison of T-DFSA with optimal FSA and TA protocols under varying packet arrival rates ($\lambda$) for $N = \{50, 100\}$.}
	\label{fig2}
	\vspace{-0.3em}
\end{figure}

In \figurename~\ref{fig2}, the AAG is plotted against the varying packet arrival rate $(\lambda)$ for T-DFSA with $w_{min} \!=\! 1$ (as in \eqref{eq11}), T-DFSA with optimal $w_{min}$, optimal FSA, and TSA protocols. Note that the results for optimal FSA were determined via simulation. As evident in this figure, T-DFSA performs up to $65\%$ better than~optimal FSA which is apparent when operating at higher arrival rates. We see that the AAG drops in T-DFSA while remaining constant in optimal FSA, especially for $\lambda$ values greater than $0.1$. This is due to the fact that the increased~number of contentions in T-DFSA under high load are classified according to the age-gains and are managed to occur in different frames.~Contrarily, increasing $\lambda$ (which is akin to adding more newly arriving updates) in optimal FSA does not yield any advantage because of the increase in collisions. Moreover, the T-DFSA with $w_{min} \!=\! 1$ is~the optimal choice when $\lambda$ exceeds a threshold (which is $0.3$ for $N \!=\! 50$ and $0.5$ for $N \!=\! 100$), where this threshold is lower at lesser values of $N$. Recalling the discussion in Section~\ref{sec5}, this happens because the variance of the age-gains diminishes with~$\lambda$ and increases with $N$, leading to more accurate estimations of $\hat{n}_t^a$. Also note that the T-DFSA with optimal $w_{min}$ outperforms TA, especially at higher generation rates. This is because, unlike TA, which has a fixed threshold, T-DFSA adjusts~its threshold in every frame based on the status of the network, such as~the number of backlogged nodes and their age-gains. As such, collisions are managed more effectively, notably at higher arrival rates.

For T-DFSA with optimum $w_{min}$, and the age-based thinning SAT and AAT protocols, the~normalized AAG (i.e., AAG~divided by $N$) is plotted against the packet generation rate in \figurename~\ref{fig3}.~As detailed in \cite{Chen2022}, the nodes in each slot under the SAT policy discard freshly generated updates with age-gains below an~optimized threshold. Based on the estimated backlogged nodes, each node then transmits its update with probability $p_b$. On~the other hand, AAT uses feedback data on the slots' collision state to~adaptively adjust the indicated threshold in each slot. \figurename~\ref{fig3} shows how T-DFSA~outperforms SAT and AAT as $\lambda$ rises, though its performance with respect to AAT decreases when the arrival rate becomes very close to $1$. The observed superiority of T-DFSA implies that it resolves collisions more effectively for the benefit of age improvement. This happens because the age-gain threshold in AAT (and SAT) is chosen so that the effective generation rate equals $1/e$, which yields the maximum achievable throughput in SA. In this manner, while the throughput is maintained at its optimum~point, it is the competing nodes that possess larger age-gains. However, T-DFSA has two degrees of freedom due to the dynamic structure of its frames. First, to maximize the throughput, the frame length can be set to be equal to the number of active nodes.~Second, by adjusting $\Gamma_t$ so that only nodes with the highest age-gain can compete, the total number of active nodes can be determined. In fact, T-DFSA achieves superior AAG since the nodes are more finely categorized for contention in frames based on their age-gains, while maintaining the optimal throughput. Moreover, as $\lambda$ approaches unity, the variability in arrival patterns decreases. Thus, T-DFSA and AAT, which both employ adaptive thresholds, exhibit similar behavior under such conditions. It is also worth noting that, in contrast to AAT, the AP is responsible for determining the dynamic thresholds in T-DFSA rather than the nodes that may be power limited. Moreover, the estimates are updated at the start of each frame rather than each slot (as in AAT), leading to lesser power utilization.
\begin{figure}[!t]
	\centering
 \includegraphics[width=0.6\columnwidth]{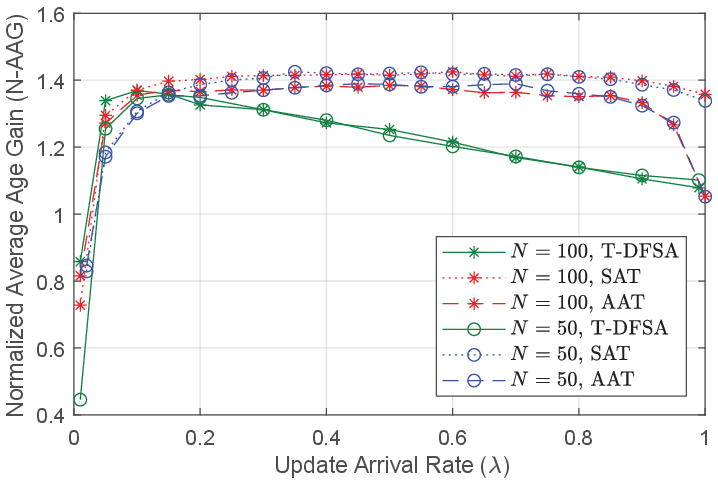}
	\vspace{-0.8em}
	\caption{Normalized AAG comparison of T-DFSA with age-based thinning protocols under varying packet arrival rates ($\lambda$) for $N = \{50, 100\}$.}
	\label{fig3}
	\vspace{-0.4em}
\end{figure}

\figurename~\ref{fig4} compares the AAG efficacy of the network under T-DFSA and TA with respect to $w_{min}$~for different $\lambda$ values. Under $N \!=\! 50$ and~$100$, the trends in \figurename~\ref{fig4}a and \figurename~\ref{fig4}b clearly show that $w_{min}$ reaches the optimal points of $1$, $2$, and $3$ as $\lambda$ drops from $0.6$~to $0.1$, respectively. Note that the AoI of a generic node can~be modeled by a geometric distribution with parameter $\lambda$ and variance $(1-\lambda)/\lambda^2$. In light of this, when the packet generation rate is low, the variation in node ages is large, resulting in high variances of age-gain as well as estimates of $\hat{f}_{t+1}^a$, derived in \eqref{eq33}. As a natural consequence of this, $w_{min} = 1$ may result in either entirely empty or collision slots. The opportune choice is to have $w_{min} \!=\! 1$ since the estimations $\hat{n}_t^a$ have smaller~variations at higher $\lambda$ rates and~are thus, more accurate. We also note that the AAG drops with $\lambda$ for $w_{min} \leq 3$ under T-DFSA, implying that fresh update generations are effectively being transmitted to reduce the AAoI. Higher values of $w_{min}$, however, compel T-DFSA to maintain a minimal frame length and prevent it from successfully classifying the available age-gains by altering $\Gamma_t$. Thereby, increasing $\lambda$ from $0.3$ to $0.6$ does not impact the AAG significantly. Interestingly, even at greater values of $w_{min}$, T-DFSA performs reasonably better than optimal FSA. Indeed, optimal FSA is shown to deliver AAGs~of $160$ and $320$ for $N= 50$ and $100$, respectively, which are much larger than the comparable AAGs of T-DFSA when $w_{min} \leq 15$.
\begin{figure}[t]
	\centering
	\includegraphics[width=0.6\columnwidth]{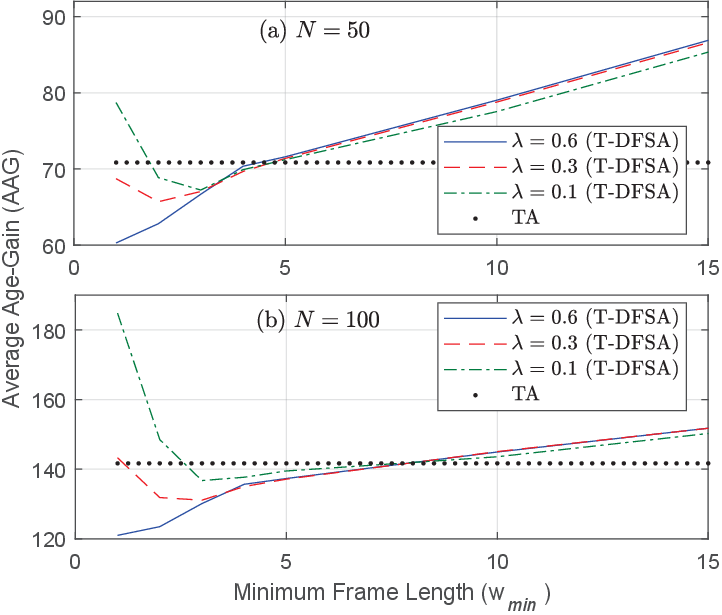}
	\vspace{-0.7em}
	\caption{AAG comparison of T-DFSA with TA protocol under varying $w_{min}$ values for $N = \{50, 100\}$ and $\lambda = \{0.1, 0.3, 0.6\}$.}
	\label{fig4}
	\vspace{-0.3em}
\end{figure}

In \figurename~\ref{fig5}, the frame length distribution is plotted for various values of $N$ and $\lambda$ while accounting for the optimal $w_{min}$. Note that the minimum frame length is $w_{min} \!+\! 1$. As an example, for $N \!=\! 100$ and $\lambda \!=\! 0.3$, we have $w_{min} \!=\! 3$ and thus, a~minimum frame length of $4$. We observe that $w_{min}$ reduces with $\lambda$ and grows with $N$. For instance, by fixing $\lambda \!=\! 0.8$, $w_{min}$ values for $N \!=\! 300$ and $N \!=\! 100$ are $1$ and $2$, respectively. The figure also discriminates between the variance in frame~lengths at different update generation rates. Particularly, the variance in frame length is greater at lower generation rates, such as $\lambda=0.05$, whereas it decreases at higher $\lambda$ values.
\begin{figure}[t]
	\centering
	\includegraphics[width=0.6\columnwidth]{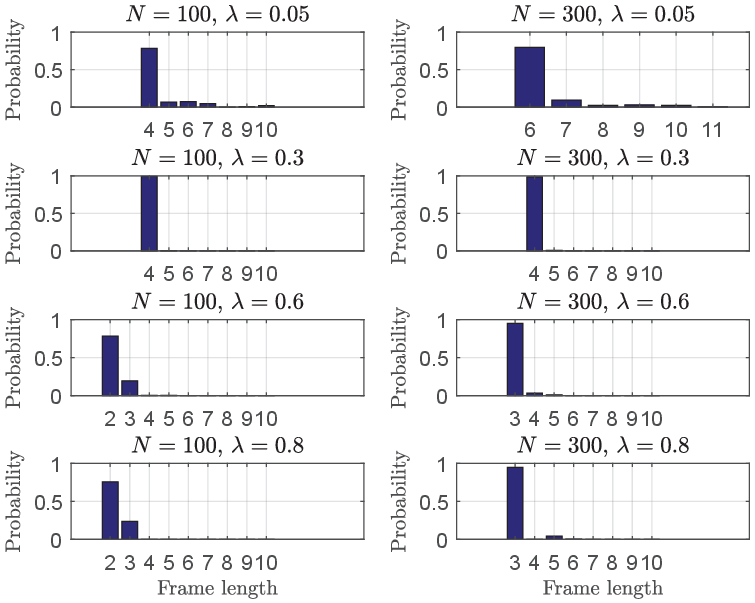}
	\vspace{-0.7em}
	\caption{Probability mass function of the frame length in T-DFSA for $N = \{100, 300\}$ and $\lambda = \{0.05, 0.3, 0.6, 0.8\}$.}
	\label{fig5}
\end{figure}
\begin{figure}[t]
	\centering
	\includegraphics[width=0.6\columnwidth]{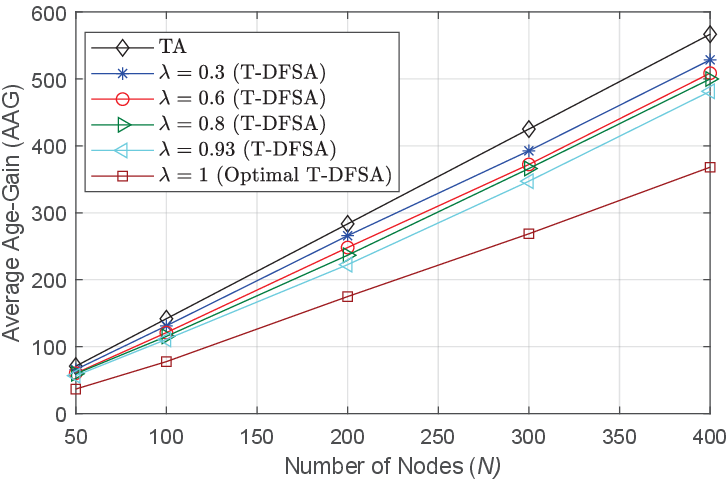}
	\vspace{-0.7em}
	\caption{AAG comparison of T-DFSA with TA protocol under varying number of source nodes ($N$) for $\lambda = \{0.3, 0.6, 0.8, 0.93, 1\}$.}
	\label{fig6}
    \vspace{-0.4em}
\end{figure}

\figurename~\ref{fig6} evaluates the effect of the number of source nodes that generate updates at different rates~on the AAG under T-DFSA with optimal $w_{min}$. Not only does the AAG increase linearly with $N$, but~the simulation results also suggest the slopes of the plotted lines decrease with $\lambda$. This implies that as~$\lambda$ rises, the estimations of $\hat{n}_t^a$ get more accurate. At $\lambda \!=\! 1$, the AP can accurately detect $n_t^a$ since the ages of the nodes are known to be one. In theory, one might expect the performance to~resemble that of an ideal T-DFSA, where the age-gains are assumed to be known at the AP. Under such load,~however, the practical T-DFSA, which relies on $\max_i y_t^i$ instead of individual $y_t^i$~values for all nodes, does not achieve the same performance as the ideal T-DFSA. On that premise, T-DFSA yet outperforms TA the most, under the high rate of $\lambda \!=\! 0.93$, by significantly lowering the AAG ($\sim 20\%$ improvement).
\begin{figure}[t]
	\centering
	\includegraphics[width=0.6\columnwidth]{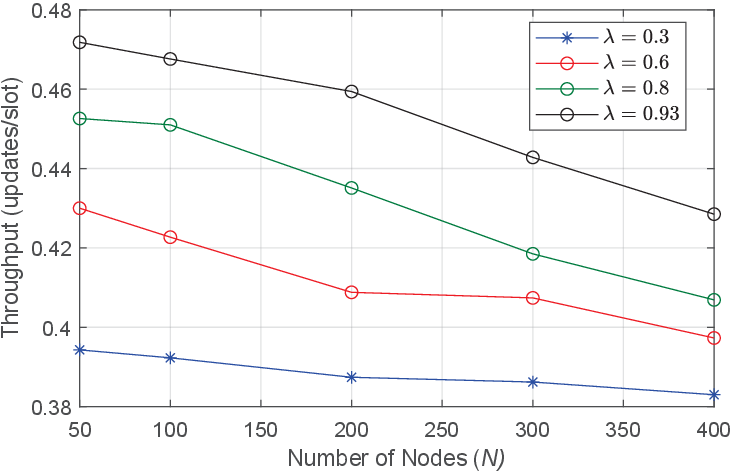}
	\vspace{-0.7em}
	\caption{Throughput of T-DFSA with optimal $w_{min}$ under varying number of source nodes ($N$) for $\lambda = \{0.3, 0.6, 0.8, 0.93\}$.}
	\label{fig7}
\end{figure}

\figurename~\ref{fig7} depicts the throughput of T-DFSA at optimal $w_{min}$~in terms of $N$ achieved under different $\lambda$ values. Notice that~the overall throughput decreases with both, $N$ and $\lambda$. This is~underpinned by the fact that the number of backlogged nodes and thus, the diversity of age-gains increase in both cases. As a result, a typical backlogged node waits longer to transmit because T-DFSA takes more time to assign frames to nodes with higher age-gains, which eventually limits the throughput. The throughput under T-DFSA is also shown to be higher than the $e^{-1}$ optimal throughput attainable with the FSA protocol. This is plainly because $e^{-1}$ is derived assuming that a relatively~large number of source nodes ($N > 10$) are competing in each frame. However with T-DFSA, the frame lengths are short, resulting in improved throughput. Albeit, if we impose the minimum~frame length constraint as indicated in Section~\ref{sec5}, the throughput achievable under T-DFSA is reduced.
\begin{figure}[!t]
  \centering
  \subfloat[N-AAoI versus $ 0.1 \leq N\lambda \leq 1$.]{
    \includegraphics[width=0.6\columnwidth]{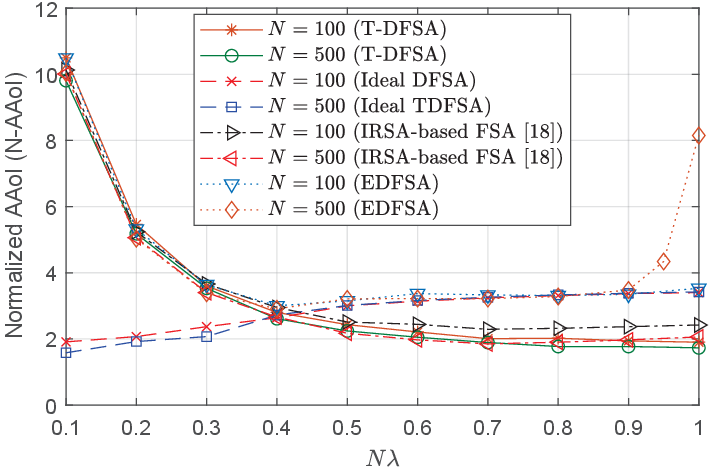}
    \label{fig:c1}
  }
  
  \subfloat[N-AAoI versus $ 0.01 \leq \lambda \leq 1$.]{
    \includegraphics[width=0.6\columnwidth]{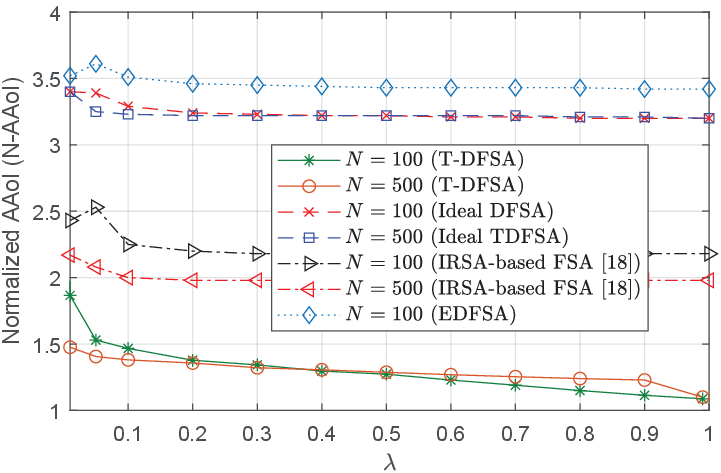}
    \label{fig:c2}
  }
  \caption{N-AAoI comparison of T-DFSA with ideal DFSA, IRSA-based FSA, and EDFSA protocols under varying packet arrival rates ($\lambda$).}
  \label{fig:compare}
  \vspace{-0.3em}
\end{figure}

Finally, \figurename~\ref{fig:compare} and Table~\ref{tab:c} provide a comprehensive comparison of the normalized AAoI (N-AAoI) for the network, which is calculated by dividing the AAoI by $N$, for T-DFSA, ideal DFSA (i.e., DFSA assuming perfect knowledge of $n_t$), IRSA-based FSA \cite{Andrea2021}, and enhanced DFSA (EDFSA) \cite{lee2005enhanced} protocols. In EDFSA, the frame size dynamically increases with the number of backlogged nodes until reaching 256, after which it remains fixed, permitting only a subset of backlogged nodes to transmit in each frame based on their backlog status. 
 Moreover, note that in the case of IRSA-based FSA, the frame size is fixed, and the N-AAoIs are computed considering the frame size that minimizes the N-AAoI while enabling the successful implementation of SIC. \figurename~\ref{fig:c1} and \figurename~\ref{fig:c2} are dedicated to different ranges of $\lambda$ values: the former covers scenarios where $N \!\lambda\! \leq 1$, and the latter focuses on $\lambda$ values~greater than or equal to $0.01$. These figures cover all values of $N\lambda$ for a network with $N \!=\! 100$ and for cases where $N\lambda \!\leq\! 1$ and $N\lambda \!\geq\! 5$ for $N \!=\! 500$. Furthermore, Table~\ref{tab:c} provides N-AAoI data for $1 \leq N\lambda \leq 5$ with $N=500$. \figurename~\ref{fig:c1} shows that T-DFSA and IRSA-based FSA outperform ideal DFSA and EDFSA when $N\lambda \geq 0.4$, with T-DFSA exhibiting superior performance compared to IRSA-based FSA at these values. Additionally, \figurename~\ref{fig:c2} and Table~\ref{tab:c} highlight that the performance superiority of T-DFSA over ideal DFSA, IRSA-based FSA, and EDFSA becomes more pronounced as $N\lambda$ exceeds $1$. It is noteworthy that when $N=500$, EDFSA becomes unstable at $N\lambda>1 $, and thus, the relevant data is not shown in \figurename~\ref{fig:c2} and Table~\ref{tab:c}. With higher arrival rates to the network, there are more fresh updates available. T-DFSA leverages this opportunity by allowing only the fresher updates with larger age-gains to compete in each frame, resulting in a higher chance of successful transmission.
\begin{table}[!t]
    \centering
    \caption{N-AAoI at $N=500$ and $2 \leq N\lambda \leq 5$ under ideal DFSA, IRSA-based FSA, and T-DFSA protocols.}
    \label{tab:c}
    \begin{tabular}{|l||c|c|c|c|c|}
        \hline
        $N\lambda$ & 1 & 2 & 3 & 4 & 5 \\ \hline\hline
        Ideal DFSA & 3.42 & 3.51 & 3.47 & 3.44 & 3.40 \\ \hline
        IRSA-based FSA [17] & 2.06 & 3.0 & 2.28 & 2.21 & 2.17 \\ \hline
        T-DFSA & 1.73 & 1.53 & 1.50 & 1.48 & 1.47 \\ \hline
    \end{tabular}
    \vspace{-0.2em}
\end{table}

\figurename~\ref{fig:comp} shows the normalized complexity per slot as a function of the packet generation rate for different values of $N$. The normalized complexity is derived by summing the complexities across~all frames and then dividing by the total number of slots. According to Remark~\ref{rem2}, $\hat{f}^a_{t}$ is computed for~all $a \leq \max y^i_{t-1}+w_{t-1}$ in Step 3, and after further truncation in Step 4, $\hat{n}^a_{t}=N\hat{f}^a_{t}$ is used in \eqref{eq11} to determine $\Gamma_t$. In \eqref{eq11}, $\Gamma_t$ is set to the largest $\Gamma'$ that satisfies the inequality $ \sum_{a \in \hat{\mathcal{G}}_t: a \geq \Gamma'} \hat{n}_t^a \geq w_{min}$, which can be rewritten as:\vspace{-0.2em}
\begin{equation}
	\sum_{a \in \hat{\mathcal{G}}_t: \Gamma' \leq a < a'} \hat{n}_t^a+\sum_{a \in \hat{\mathcal{G}}_t: a \geq a' } \hat{n}_t^a \geq w_{min},
	\label{eq_last}
 \vspace{-0.2em}
\end{equation}
where $a'$ is the smallest age-gain that satisfies the inequality $ \sum_{a \in \hat{\mathcal{G}}_t: a \geq a'} \hat{n}_t^a \leq 1$. Since $w_{min} \geq 1$, $\Gamma_t$ is guaranteed to~be greater than $a'$, thus omitting the calculations in Step 3 for all values of $a \geq a'$, i.e., only the summation $\sum_{a \in \hat{\mathcal{G}}_t: a' \leq a } \hat{n}_t^a$ is sufficient. Now, the first data point on the three curves in \figurename~\ref{fig:comp} corresponds to $\lambda=0.001$. Therefore, the total arrival rate $N \lambda$ equals $0.1$ for $N \!=\! 100$, which is less than the~maximum achievable throughput of DFSA, i.e., $e^{-1}$. In this case, almost all updates are transmitted with high probability, resulting in a very small average number of backlogged nodes. Hence, we set $w_{min}=1$. Additionally, in many frames, the estimated number of backlogged nodes is less than one, thus leading to a frame length of one. As a result, the calculations in Step 3 are significantly reduced using the aforementioned technique. However, as $\lambda$ increases and $N\lambda$ surpasses the maximum throughput~of $e^{-1}$, the complexity initially increases. But as explained in Section~\ref{sec5}, the value of $\max_i y^i_t$ decreases with $\lambda$, resulting in a subsequent reduction in complexity. Undoubtedly, the figure also shows that the complexity increases with $N$.\vspace{-0.2em}
\begin{figure}[!t]
	\centering
	\includegraphics[width=0.6\columnwidth]{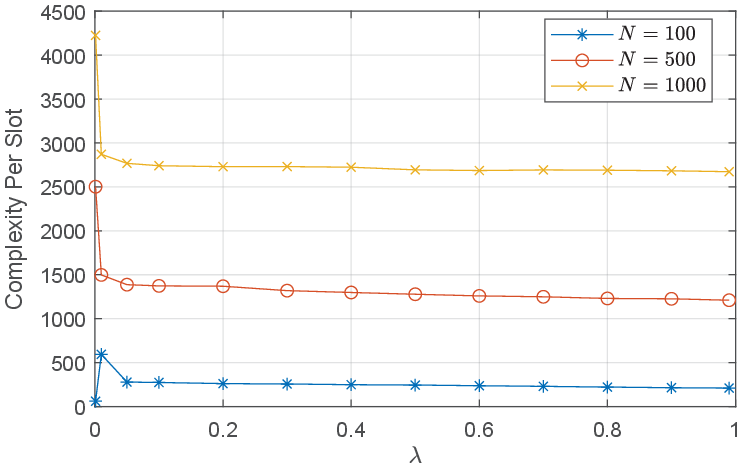}
	\vspace{-0.7em}
	\caption{Normalized complexity per slot versus arrival rate ($\lambda$) for $N= \{100, 500, 1000\}$.}
	\label{fig:comp}
	\vspace{-0.3em}
\end{figure}

\section{Conclusion}
\label{sec7}
\fontdimen2\font=0.54ex
We developed an age-aware \emph{threshold-based} DFSA protocol for large-scale IoT networks with stochastic update arrivals.~The AP broadcasts the frame length and a threshold at the start~of each frame in the proposed T-DFSA, designating nodes with age-gains larger than the threshold as suitable nodes to transmit. The average number of nodes with various age-gains at the beginning of the frames is estimated using a four step approach by the AP in order to determine the threshold and frame length. The information observed throughout the frame, the age-gains of successful nodes, and AoIs at the AP are all used in this procedure, together with any available information on the node arrival rate. The presented analysis revealed that T-DFSA becomes stable provided the initial condition of the system is sufficiently diverse. Numerical results showed that T-DFSA significantly improves the~age performance, with gains of up to $65\%$ compared to its optimal FSA counterpart. Furthermore, T-DFSA demonstrated to be a practical solution for power-limited networks by surpassing the TA,~SAT, and AAT baseline protocols in terms of AAG. Extending the analysis of AoI in T-DFSA to settings that involve nodes with different sampling behaviors and age requirements using machine learning driven techniques is worthy of further investigation.\vspace{-0.4em}


\appendices
\section{Proof of Lemma~\ref{lemma1}}
\label{appA}
\fontdimen2\font=0.52ex
We first show that $w_t = n_t$ maximizes $\bar{R}_t$ for a given $\Gamma_t$ value. For now, we assume that $w_t$ is a continuous random variable. Differentiating \eqref{eq8} with respect to $w_t$ yields:\vspace{-0.2em}
\begin{equation}
	\frac{\partial \bar{R}_t}{\partial w_t} = \frac{\sum_{a=\Gamma_t}^M a n_t^a}{N w_t^2}\left(1 - \frac{1}{w_t}\right)^{n_t - 2}\left(\frac{n_t}{w_t} - 1\right).
	\label{eq40}
    \vspace{-0.2em}
\end{equation}
By equating $\partial \bar{R}_t/\partial w_t = 0$, we can conclude that, for a given $\Gamma_t$, setting $w_t = n_t$ maximizes $\bar{R}_t$.~Note that $w_t=1$ also satisfies $\partial \bar{R}_t/\partial w_t = 0$, but it minimizes $\bar{R}_t$, thus resulting in $\bar{R}_t=-1$. Thus, since $w_t=n_t$ maximizes $\bar{R}_t$ over all continuous values of $w_t$, it is also the optimal discrete value. Next, we optimize the value of $\Gamma_t$ by setting $w_t = n_t = \sum_{a=\Gamma_t}^M n_t^a$ in \eqref{eq8} to obtain:\vspace{-0.2em}
\begin{equation}
	\bar{R}_t^*(\Gamma_t) = \underbrace{\left(1 - \frac{1}{\sum_{a=\Gamma_t}^M n_t^a}\right)^{\sum_{a=\Gamma_t}^M n_t^a - 1}}_\textsl{x} \underbrace{\frac{\sum_{a=\Gamma_t}^M a n_t^a}{N \sum_{a=\Gamma_t}^M n_t^a}}_\textsl{y} - 1,
	\label{eq41}
 \vspace{-0.2em}
\end{equation}
where $\bar{R}_t^*(\Gamma_t)$ is the maximum of $\bar{R}_t$ for a given value of $\Gamma_t$. We now show that  $\bar{R}_t^*(\Gamma_t)$ is a decreasing function of $\Gamma_t$. Note that $n_t = \sum_{a=\Gamma_t}^M n_t^a$ is a decreasing function of $\Gamma_t$ since the number of active~nodes reduces as $\Gamma_t$ grows. Alternatively, for $n_t \geq 1$, the function $(1 - 1/n_t)^{n_t-1}$ increases with $n_t$. Thus,~the term $\textsl{x}$ in \eqref{eq41} is a decreasing function of $\Gamma_t$, given that $\sum_{a=\Gamma_t}^M n_t^a \geq 1$. Moreover, the term $\textsl{y}$ in \eqref{eq41} can be written as follows, which is a convex combination of the values $\{\Gamma_t, \Gamma_t + 1, \ldots, M\}$:\vspace{-0.2em}
\begin{equation}
	\sum_{a=\Gamma_t}^M a \frac{n_t^a}{\sum_{a=\Gamma_t}^M n_t^a}.
	\label{eq42}
    \vspace{-0.2em}
\end{equation}
As a result, it attains its maximum value whenever the coefficient of the highest value of $a$ in \eqref{eq42} is equal to one. This leads to $n_t^M / \sum_{a = \Gamma_t}^M n_t^a = 1$, which transpires only when $\Gamma_t = M$. This completes the proof.\vspace{-0.5em}

\section{Proof of Lemma~\ref{lemma2}}
\label{appB}
\fontdimen2\font=0.52ex
For \eqref{eq44} to be positive, we should have $n_t^a \geq n_{t,s}^a$. Thus, the optimal values of $\hat{m}_t^a$ should satisfy $\hat{m}_t^a \geq n_{t,s}^a$ for any $a \in \mathcal{S}_t$, implying that $\mathcal{S}_t \subset \hat{\mathcal{A}}_t$. In regard to $\sum_{a \in \hat{\mathcal{A}}_t} \hat{m}_t^a = l$, we should set $\hat{m}_t^a = 0$ for $a \notin \mathcal{S}_t$ in order to maximize the probability in \eqref{eq44}. As a result of this, we will have $\hat{\mathcal{A}}_t = \mathcal{S}_t$.

Let us now consider $a,b \in \hat{\mathcal{A}}_t$. Accounting for the fact that $\hat{\mathcal{L}}_t(l)$ is the optimal solution of \eqref{eq43}, flipping $b$ to $a$ in $\hat{\mathcal{L}}_t(l)$ (i.e., changing $\hat{m}_t^a$ to $\hat{m}_t^a + 1$ and $\hat{m}_t^b$ to $\hat{m}_t^b - 1$) should reduce the probability in \eqref{eq44}. Expressly, the difference between the probability before and after flipping can be stated as:\vspace{-0.2em}
\begin{equation}
	\frac{\prod_{c \in \hat{\mathcal{A}}_t, c \neq a,b} \binom{\hat{m}_t^c}{n_{t,s}^c}}{\binom{l}{n_S}}\left(\!\binom{\hat{m}_t^a \!+\! 1}{n_{t,s}^a} \binom{\hat{m}_t^b - 1}{n_{t,s}^b} - \binom{\hat{m}_t^a}{n_{t,s}^a} \binom{\hat{m}_t^b}{n_{t,s}^b}\!\right).
	\label{eq45}
 \vspace{-0.2em}
\end{equation}
By letting \eqref{eq45} to be less than or equal to zero, \eqref{eq24} can be obtained straightforwardly. This completes the proof.\vspace{-0.3em}

\section{Proof of Lemma~\ref{lemma3}}
\label{appC}
\fontdimen2\font=0.52ex
Suppose $\hat{m}_t^x / n_{t,s}^x = k'_x + r_x / n_{t,s}^x$, $\forall x \in \hat{\mathcal{A}}_t$, we have $k'_x, r_x \in \mathbb{Z}^+$ and $0 \leq r_x < n_{t,s}^x$. Using \eqref{eq44}, for any $a,b \in \hat{\mathcal{A}}_t$, we have:\vspace{-0.2em}
\begin{equation}
	\frac{r_a}{n_{t,s}^a} - \frac{r_b + 1}{n_{t,s}^b} \leq k'_b - k'_a \leq \frac{r_a + 1}{n_{t,s}^a} - \frac{r_b}{n_{t,s}^b}.
	\label{eq46}
    \vspace{-0.2em}
\end{equation}
Depending on the values taken by $r_x$, the following two cases may occur:
\begin{itemize}
	\item \emph{Case I:} When $r_x = 0$, $\forall x \in \hat{\mathcal{A}}_t$. By assigning $r_a = r_b = 0$ in  \eqref{eq46} along with the fact~that $n_{t,s}^a, n_{t,s}^b \geq 1$, we can conclude that $k'_b = k'_a$ (i.e., $\forall x, \hat{m}_t^x / n_{t,s}^x = k$) and thus, $k' = l / n_S$.
	\item \emph{Case II:} When $0 < r_a < n_{t,s}^a$ for some $a \in \hat{\mathcal{A}}_t$. From~\eqref{eq46} and the assertion that $(1+r_x) / n_{t,s}^x \leq 1$, we conclude that $\forall b \in \hat{\mathcal{A}}_t$, $|k'_b - k'_a| \leq 1$. Simply put, either $k'_b = k'_a$ or $k'_b = k'_a + 1$ holds, where the latter transpires if and only if $r_a = n_{t,s}^a - 1$ and $r_b = 0$. Notably, the case $k'_b = k'_a - 1$ never occurs since it mandates $r_a = 0$ and $r_b = n_{t,s}^b - 1$, which contradicts our assumption (i.e., $r_a > 0$). Subsequently, $\hat{m}_t^b, \forall b$, can be~expressed either as $k'_a n_{t,s}^b + r_b$ or $(k'_a + 1) n_{t,s}^b$. This is equivalent to writing $k'_a n_{t,s}^b + r'_b$, where $0 \leq r'_b \leq n_{t,s}^b$. Summing over $\hat{m}_t^x, \forall x \in \hat{\mathcal{A}}_t$, we attain $l = k'_a n_S + \sum_{x \in \hat{\mathcal{A}}_t} r'_x$, where $x=a$ and $r'_x = r_a$. Since for $a$, $r'_x = r_a < n_{t,s}^a$, we also have $\sum_{x \in \hat{\mathcal{A}}_t} r'_x < n_S$, which eventually leads to $k'_a = \lfloor l / n_S \rfloor$.
\end{itemize}\vspace{-0.4em}

\section{Proof of Proposition~\ref{prop1}}
\label{appE}
\fontdimen2\font=0.52ex
Let $x_{t^+}^i$, $y_{t^+}^i$, and $g_{t^+}^i$ denote, respectively, the packet age of node $i$, the age at the AP with respect~to node $i$, and the age-gain of node $i$ at the end of frame $t$, without considering the effect of new arrivals within the frame. Contrarily, $x_{t+1}^i$, $y_{t+1}^i$, and $g_{t+1}^i$, which are the corresponding values at the inception of frame $t+1$, take in the effect of new update arrivals. Accordingly, $\hat{f}_{t+1}^a$ can be written as:\vspace{-0.2em}
\begin{align}
	\hat{f}_{t+1}^a &= \Pr(y_{t+1}^i - x_{t+1}^i = a) \nonumber \\
	&= \sum_{b=0}^a \Pr(y_{t+1}^i - x_{t^+}^i = b)  \Pr(x_{t^+\!}^i - x_{t+1}^i = a - b \mid a, b) \nonumber \\
	&= \sum_{b=0}^a f_{t^+}^b \underbrace{\Pr(x_{t^+\!}^i - x_{t+1}^i = a - b \mid a, b)}_\textsl{x},
	\label{eq47}
    \vspace{-0.3em}
\end{align}
where the last equality is inferred noting that $y_{t+1}^i \!=\! y_{t^+}^i$ and $\hat{f}_{t^+}^b \!=\! \Pr(g_{t^+}^i \!=\! y_{t^+}^i - x_{t^+}^i \!=\! b)$. Additionally, because the age-gain is always boosted by new arrivals, we have $g_{t^+}^i \leq g_{t+1}^i$, which means that the values $b \leq a$ are taken into consideration in \eqref{eq47}. The probability term $\textsl{x}$ in \eqref{eq47} can further be expressed~as:\vspace{-0.2em}
\begin{equation}
  \begin{cases}
    (1 - \lambda)^{w_t}, & \text{ if $b = a$} \\
    \sum_c \underbrace{\Pr(x_{t+1}^i = c)}_\textsl{y} \underbrace{\Pr(x_{t^+}^i = c + a - b \mid a, b, c)}_\textsl{z},       & \text{ if $b < a$}.
  \end{cases}
  \label{eq48}
\end{equation}
The first case in \eqref{eq48} indicates that the age of node $i$ does not alter between $t^+$ and the onset of frame $t \!+\! 1$, as long as no~update is generated in frame $t$, with probability $(1 \!-\! \lambda)^{w_t}$. In the second case of \eqref{eq48}, the term $\textsl{y}$ represents the probability that the latest update of node $i$ is generated at time slot $k_{t+1} - c$ as per \eqref{eq1}, and is given by $p_c \triangleq \lambda(1 - \lambda)^{c-1}$. Likewise, the term $\textsl{z}$ in \eqref{eq48} is equivalent to $\Pr(x_t^i = c + a -b - w_t \mid a, b, c)$. By defining $h \triangleq c + a - b - w_t$, the term $\textsl{z}$ can be written as $\Pr(x_t^i = h \mid a, b, c)$ and refers to the probability that the latest update at node $i$ prior to slot $k_t$ is generated in slot $k_t - h$. Leaning on the following facts, we now obtain the bounds for $h$ and $c$:
\begin{enumerate}
	\item[(i)] $x_t^i \geq 1$, where $x_t^i = 1$ occurs only if $u^i(k_t - 1) = 1$. Hence, in \eqref{eq43}, we have $h \geq 1$.
	\item[(ii)] $x_t^i \leq k_t + x_0^i, \forall i$, where the equality holds when no updates are generated by node $i$ since $k_0 = 0$. Hence, $h \leq k_t + \max_i x_0^i$.
	\item[(iii)] $g_t^i \geq g_{t^+}^i$, where we have $g_t^i = y_t^i - x_t^i \geq b$ given that $g_{t^+}^i = b$. This implies that $x_t^i \leq y_t^i - b$ and consequently, $h \leq \max_i y_t^i - b$.
\end{enumerate}
From fact (i), we easily conclude that $h_{min} = 1$. Then, from~$h \geq 1$, we get $c \geq b - a + w_t + 1$. Along with $c \geq 1$, we have $c_{min} = (b - a + w_t)^+ + 1$. Facts (ii) and (iii) reveal that $h \leq h_{max} = \min\{\max_i x_0^i + k_t, \max_i y_t^i - b\}$. From this, $c \leq h_{max} - a + b + w_t$ can be directly inferred, which accompanied by $1 \leq c \leq w_t$, results in $c_{max} = \max \{(h_{max} - a + b)^- + w_t, 1\}$. To this end, $\Pr(x_t^i = h \mid a, b, c)$ in \eqref{eq48} can be stated as:\vspace{-0.2em}
\begin{align}
	\tilde{p}_h \triangleq \Pr(x_t^i = h \mid a, b, c) &= \Pr(x_t^i = h \mid 1 \leq h \leq h_{max}) \nonumber \\
	&= \frac{\Pr(x_t^i = h, 1 \leq h \leq h_{max})}{\Pr(1 \leq h \leq h_{max})} \nonumber \\
	&= \frac{\lambda(1 - \lambda)^h}{\sum_{h=1}^{h_{max}} \lambda(1 - \lambda)^h}.
	\label{eq49}
    \vspace{-0.2em}
\end{align}
Plugging \eqref{eq49} into \eqref{eq48} completes the proof.\vspace{-0.4em}



%
%

\ifCLASSOPTIONcaptionsoff
  \newpage
\fi



\bibliographystyle{IEEEtran}
\bibliography{IEEEabrv,comref.bib}

\begin{thebibliography}{10}
\providecommand{\url}[1]{#1}
\csname url@samestyle\endcsname
\providecommand{\newblock}{\relax}
\providecommand{\bibinfo}[2]{#2}
\providecommand{\BIBentrySTDinterwordspacing}{\spaceskip=0pt\relax}
\providecommand{\BIBentryALTinterwordstretchfactor}{4}
\providecommand{\BIBentryALTinterwordspacing}{\spaceskip=\fontdimen2\font plus
\BIBentryALTinterwordstretchfactor\fontdimen3\font minus
  \fontdimen4\font\relax}
\providecommand{\BIBforeignlanguage}[2]{{%
\expandafter\ifx\csname l@#1\endcsname\relax
\typeout{** WARNING: IEEEtran.bst: No hyphenation pattern has been}%
\typeout{** loaded for the language `#1'. Using the pattern for}%
\typeout{** the default language instead.}%
\else
\language=\csname l@#1\endcsname
\fi
#2}}
\providecommand{\BIBdecl}{\relax}
\BIBdecl

\bibitem{Rita2016}
M.~R. Palattella, M.~Dohler, A.~Grieco, G.~Rizzo, J.~Torsner, T.~Engel, and
  L.~Ladid, ``Internet of {T}hings in the {5G} era: Enablers, architecture, and
  business models,'' \emph{IEEE J. Sel. Areas Commun.}, vol.~34, no.~3, pp.
  510--527, Mar. 2016.

\bibitem{Elmagid2019}
M.~A. Abd-Elmagid, N.~Pappas, and H.~S. Dhillon, ``On the role of age of
  information in the {I}nternet of {T}hings,'' \emph{IEEE Commun. Mag.},
  vol.~57, no.~12, pp. 72--77, Dec. 2019.

\bibitem{Chiariotti2022}
F.~Chiariotti, O.~Vikhrova, B.~Soret, and P.~Popovski, ``Age of information in
  multihop connections with tributary traffic and no preemption,'' \emph{IEEE
  Trans. Commun.}, vol.~70, no.~10, pp. 6718--6733, Oct. 2022.

\bibitem{Kaul2012}
S.~Kaul, R.~Yates, and M.~Gruteser, ``Real-time status: How often should one
  update?'' in \emph{Proc. IEEE Int. Conf. Comput. Commun. (INFOCOM)}, Mar.
  2012, pp. 2731--2735.

\bibitem{Kumar2023}
M.~S. Kumar, A.~Dadlani, M.~Moradian, A.~Khonsari, and T.~A. Tsiftsis, ``On the
  age of status updates in unreliable multi-source {M/G/1} queueing systems,''
  \emph{IEEE Commun. Lett.}, vol.~27, no.~2, pp. 751--755, Feb. 2023.

\bibitem{Yates2021}
R.~D. Yates, Y.~Sun, D.~R. Brown, S.~K. Kaul, E.~Modiano, and S.~Ulukus, ``Age
  of information: An introduction and survey,'' \emph{IEEE J. Sel. Areas
  Commun.}, vol.~39, no.~5, pp. 1183--1210, May 2021.

\bibitem{Ghavimi2015}
F.~Ghavimi and H.-H. Chen, ``{M2M} communications in {3GPP} {LTE}/{LTE-A}
  networks: Architectures, service requirements, challenges, and
  applications,'' \emph{IEEE Commun. Surveys \& Tutor.}, vol.~17, no.~2, pp.
  525--549, 2015.

\bibitem{Atabay2020}
D.~C. Atabay, E.~Uysal, and O.~Kaya, ``Improving age of information in random
  access channels,'' in \emph{Proc. IEEE Int. Conf. Comput. Commun. Workshops
  (INFOCOM WKSHPS)}, Jul. 2020, pp. 912--917.

\bibitem{Tahir2021}
O.~T. Yavascan and E.~Uysal, ``Analysis of slotted {ALOHA} with an age
  threshold,'' \emph{IEEE J. Sel. Areas Commun.}, vol.~39, no.~5, pp.
  1456--1470, May 2021.

\bibitem{Chen2022}
X.~Chen, K.~Gatsis, H.~Hassani, and S.~S. Bidokhti, ``Age of information in
  random access channels,'' \emph{IEEE Trans. Inf. Theory}, vol.~68, no.~10,
  pp. 6548--6568, Oct. 2022.

\bibitem{Ahmetoglu2022}
M.~Ahmetoglu, O.~T. Yavascan, and E.~Uysal, ``Mi{STA}: An age-optimized slotted
  {ALOHA} protocol,'' \emph{IEEE Internet Things J.}, vol.~9, no.~17, pp.
  15\,484--15\,496, Sep. 2022.

\bibitem{Ajinkya2015}
A.~Rajandekar and B.~Sikdar, ``A survey of {MAC} layer issues and protocols for
  machine-to-machine communications,'' \emph{IEEE Internet Things J.}, vol.~2,
  no.~2, pp. 175--186, 2015.

\bibitem{Yu2021}
J.~Yu, P.~Zhang, L.~Chen, J.~Liu, R.~Zhang, K.~Wang, and J.~An, ``Stabilizing
  frame slotted {ALOHA}-based {I}o{T} systems: A geometric ergodicity
  perspective,'' \emph{IEEE J. Sel. Areas Commun.}, vol.~39, no.~3, pp.
  714--725, Mar. 2021.

\bibitem{wang2023age}
J.~Wang, J.~Yu, X.~Chen, L.~Chen, C.~Qiu, and J.~An, ``Age of information for
  frame slotted {ALOHA},'' \emph{IEEE Trans. Commun.}, vol.~71, no.~4, pp.
  2121--2135, Apr. 2023.

\bibitem{yue2023age}
Z.~Yue, H.~H. Yang, M.~Zhang, and N.~Pappas, ``Age of information under frame
  slotted {ALOHA}-based status updating protocol,'' \emph{IEEE J. Sel. Areas
  Commun.}, vol.~41, no.~7, pp. 2071--2089, Jul. 2023.

\bibitem{huang2023distributed}
W.-C. Huang, Y.-C. Huang, K.-Y. Lin, S.-L. Shieh, and P.-N. Chen, ``Distributed
  scheduling for status update with heterogeneous services under the {IRSA}
  protocol,'' \emph{IEEE Trans. Veh. Technol.}, 2023.

\bibitem{Saha2021}
S.~Saha, V.~B. Sukumaran, and C.~R. Murthy, ``On the minimum average age of
  information in {IRSA} for grant-free m{MTC},'' \emph{IEEE J. Sel. Areas
  Commun.}, vol.~39, no.~5, pp. 1441--1455, May 2021.

\bibitem{Andrea2021}
A.~Munari, ``Modern random access: An age of information perspective on
  irregular repetition slotted {ALOHA},'' \emph{IEEE Trans. Commun.}, vol.~69,
  no.~6, pp. 3572--3585, Jun. 2021.

\bibitem{Huang2022}
Y.~Huang, J.~Jiao, S.~Wu, R.~Lu, and Q.~Zhang, ``Graph-based spatially coupled
  {IRSA} random access for age-critical grant-free massive access,'' in
  \emph{Proc. IEEE Int. Conf. .Commun. (ICC)}, May 2022, pp. 1986--1991.

\bibitem{Li2022}
Y.~Li and K.-W. Chin, ``Energy-aware irregular slotted {ALOHA} methods for
  wireless-powered {I}o{T} networks,'' \emph{IEEE Internet Things J.}, vol.~9,
  no.~14, pp. 11\,784--11\,795, Jul. 2022.

\bibitem{stefanovic2012frameless}
C.~Stefanovic, P.~Popovski, and D.~Vukobratovic, ``Frameless {ALOHA} protocol
  for wireless networks,'' \emph{IEEE Commun. Lett.}, vol.~16, no.~12, pp.
  2087--2090, 2012.

\bibitem{huang2023age}
Y.~Huang, J.~Jiao, Y.~Wang, X.~Zhang, S.~Wu, R.~Lu, and Q.~Zhang, ``Age of
  information minimization for frameless {ALOHA} in grant-free massive
  access,'' \emph{IEEE Trans. Wirel. Commun.}, 2023.

\bibitem{munari2023dynamic}
A.~Munari, F.~L{\'a}zaro, G.~Durisi, and G.~Liva, ``The dynamic behavior of
  frameless {ALOHA}: Drift analysis, throughput, and age of information,''
  \emph{IEEE Trans. Commun.}, 2023.

\bibitem{Frits1983}
F.~Schoute, ``Dynamic frame length {ALOHA},'' \emph{IEEE Trans. Commun.},
  vol.~31, no.~4, pp. 565--568, Apr. 1983.

\bibitem{Su2019}
J.~Su, Z.~Sheng, V.~C.~M. Leung, and Y.~Chen, ``Energy efficient tag
  identification algorithms for {RFID}: Survey, motivation and new design,''
  \emph{IEEE Wireless Commun.}, vol.~26, no.~3, pp. 118--124, Jun. 2019.

\bibitem{Bae2022}
Y.~H. Bae and J.~W. Baek, ``Age of information and throughput~in random
  access-based {I}o{T} systems with periodic updating,'' \emph{IEEE Wirel.
  Commun. Lett.}, vol.~11, no.~4, pp. 821--825, Apr. 2022.

\bibitem{Kadota2021}
I.~Kadota and E.~Modiano, ``Age of information in random access networks with
  stochastic arrivals,'' in \emph{Proc. IEEE Conf. Comput. Commun. (INFOCOM)},
  May 2021, pp. 1--10.

\bibitem{Yates2017}
R.~D. Yates and S.~K. Kaul, ``Status updates over unreliable multiaccess
  channels,'' in \emph{Proc. IEEE Int. Symp. Inf. Theory (ISIT)}, June 2017,
  pp. 331--335.

\bibitem{Munari2021}
A.~Munari and G.~Liva, ``Information freshness analysis of slotted {ALOHA} in
  {G}ilbert-{E}lliot channels,'' \emph{IEEE Commun. Lett.}, vol.~25, no.~9, pp.
  2869--2873, Sep. 2021.

\bibitem{Pan2022}
H.~Pan, T.-T. Chan, J.~Li, and V.~C.~M. Leung, ``Age of information with
  collision-resolution random access,'' \emph{IEEE Trans. Veh. Technol.},
  vol.~71, no.~10, pp. 11\,295--11\,300, Oct. 2022.

\bibitem{Wang2022}
Q.~Wang and H.~Chen, ``Age of information in reservation multi-access networks
  with stochastic arrivals,'' in \emph{Proc. IEEE Int. Symp. Inform. Theory
  (ISIT)}, Aug. 2022, p. 2088–2093.

\bibitem{Yang2022}
T.~Yang, J.~Jiao, S.~Wu, R.~Lu, and Q.~Zhang, ``Grant free age-optimal random
  access protocol for satellite-based {I}nternet of {T}hings,'' \emph{IEEE
  Trans. Commun.}, vol.~70, no.~6, pp. 3947--3961, Jun. 2022.

\bibitem{Vogt2002}
H.~Vogt, ``Efficient object identification with passive {RFID} tags,'' in
  \emph{Lecture Notes in Computer Science}.\hskip 1em plus 0.5em minus
  0.4em\relax Springer Berlin Heidelberg, 2002, vol. 2414, pp. 98--113.

\bibitem{Chen2009}
W.-T. Chen, ``An accurate tag estimate method for improving the performance of
  an {RFID} anticollision algorithm based on dynamic frame length {ALOHA},''
  \emph{IEEE Trans. Autom. Sci. Eng.}, vol.~6, no.~1, pp. 9--15, Jan. 2009.

\bibitem{lee2005enhanced}
S.-R. Lee, S.-D. Joo, and C.-W. Lee, ``An enhanced dynamic framed slotted
  {ALOHA} algorithm for {RFID} tag identification,'' in \emph{Proc. IEEE
  International Conference on Mobile and Ubiquitous Systems: Networking and
  Services (MOBIQUITOUS)}, 2005, pp. 166--172.

\end{thebibliography}
%
%
%
%

%




\end{document}